\documentclass[twoside,11pt]{article}
\usepackage{jmlr2e}
\ShortHeadings{Stacking for Non-mixing Bayesian Computations}{Yao, Vehtari and Gelman}
\firstpageno{1}
\usepackage{graphicx}
 \usepackage{natbib}
\usepackage{amsfonts,epsfig}
 \usepackage{microtype}
\usepackage{mathtools}
\usepackage{bbm}
\usepackage{extarrows}
\usepackage{enumitem}
\usepackage{tabularx}
\usepackage{pifont} 
\usepackage{xcolor}
\definecolor{darkblue}{rgb}{0,0.08,0.45}
\definecolor{darkred}{RGB}{139,0,0}
\definecolor{Darkblue}{RGB}{0,0,139}
\definecolor{forestgreen}{RGB}{34,139,34}
 \definecolor{darkgreen}{RGB}{0,100,0}
 \usepackage{tikz}
 \usetikzlibrary{shapes,arrows,chains}
 \usetikzlibrary[calc]
 \usetikzlibrary{bayesnet}
\newcommand{\cmark}{\ding{51}}%
\newcommand{\xmark}{\ding{55}}%
\graphicspath{{fig/}}
\usepackage{wrapfig}
\usepackage{relsize}
\usepackage{fancyvrb}
\let\oldv\verbatim
\let\oldendv\endverbatim
\def\verbatim{\par\setbox0\vbox\bgroup\oldv}
\def\endverbatim{\oldendv\egroup\fboxsep0pt \noindent\colorbox[gray]{0.96}{\usebox0}\par}
\usepackage{amssymb}
\usepackage{subcaption}
\usepackage{amsmath}
\DeclareMathOperator{\E}{\mathbbm{E}}

\DeclareMathOperator{\Var}{\mathrm{Var}}
\DeclareMathOperator{\KL}{\mathrm{KL}}
\DeclareMathOperator{\R}{\mathbb{R}}

\DeclareMathOperator{\w}{\mathbf{w}}

\DeclarePairedDelimiter\floor{\lfloor}{\rfloor}
\newcolumntype{b}{X}
\newcolumntype{s}{>{\hsize=.75\hsize}X}

\begin{document}\sloppy
\title{Stacking for Non-mixing Bayesian Computations: \\The Curse and Blessing of Multimodal Posteriors}
\author{\name Yuling Yao \email yyao@flatironinstitute.org\\
	\addr Flatiron Institute\\
	New York, NY 10010, USA
	\AND
	\name Aki Vehtari \email aki.vehtari@aalto.fi\\
	\addr Department of Computer Science, Aalto University\\
	00076 Aalto, Finland
       \AND
\name Andrew Gelman \email gelman@stat.columbia.edu\\
	\addr Department of Statistics and of Political Science,
	Columbia University\\
New York, NY 10027, USA}

 \editor{}

\maketitle

\begin{abstract}
When working with multimodal Bayesian posterior distributions, Markov chain Monte Carlo (MCMC) algorithms have difficulty moving between modes, and default variational or mode-based approximate inferences will understate posterior uncertainty.  And, even if the most important modes can be found, it is difficult to evaluate their relative weights in the posterior.
Here we propose an approach using parallel runs of MCMC, variational, or mode-based inference to hit as many modes or separated regions as possible and then combine these using Bayesian stacking, a scalable method for constructing a weighted average of distributions.
The result from stacking efficiently samples from multimodal posterior distribution, minimizes cross validation prediction error, and represents the posterior uncertainty better than variational inference, but it is not necessarily equivalent, even asymptotically, to fully Bayesian inference.
We present theoretical consistency with an example where the stacked inference approximates the true data generating process from the misspecified model and a non-mixing sampler, from which the predictive performance is better than full Bayesian inference, hence the multimodality can be considered a blessing rather than a curse under model misspecification.
We demonstrate practical implementation in several model families: latent Dirichlet allocation,  Gaussian process regression, hierarchical regression, horseshoe variable selection, and neural networks.
\end{abstract}


\section{Introduction}\label{sec_intro}

Bayesian computation becomes difficult when posterior
distributions are multimodal or more generally metastable, that is, with high-probability regions separated by regions of low probability. Such pathology commonly arises with mixtures \citep{stephens2000dealing},  hierarchical models \citep{liu2003posterior}, and overparametrized models \citep{izmailov2021bayesian}.
It is impossible in general to compute moments analytically or to directly draw simulations,  variational and mode-based approximations can yield poor fits to the posterior \citep{yao2018yes}, and general-purpose Markov chain Monte Carlo algorithms can have problems moving between modes \citep{rudoy2006monte}. For example, an optimally tuned Hamiltonian Monte Carlo sampler for a bimodal density mixes as poorly as a random-walk Metropolis sampler \citep{mangoubi2018does}.

The extra challenge is that problems in sampling and modeling are confounded. Even if we can sample from truly multimodal distributions, the posterior multimodality signifies that the true data are unlikely to have been generated  from any single parameter in the model, so that the Bayesian posterior itself, which has to concentrate somewhere in the limit, may not be appropriate.

One way to explore a multimodal space is to run many chains of MCMC or variational inference from dispersed starting points, but then the question arises of how to combine non-mixing inferences. 
Even if all modes are found, it is difficult to compute their relative weights in the posterior distribution, as this requires integration over the posterior density within each mode.
Consider $K$ chains of parameter vectors, where the $k$-th chain contains $S_k$ draws, $(\theta_{k1}, \dots, \theta_{kS_k})$. We  consider a generalized form of Monte Carlo estimate for any integral function $h(\theta)$ from a chainwise weight $\w=( w_{1}, w_2, \ldots, w_{K})$:
\begin{equation}\label{eq_MC_final}
	\E( h(\theta) | \w)\approx  \sum_{k=1}^K \sum_{s=1}^{S_k}  w_k S_k^{-1} h(\theta_{ks}).	\end{equation}
The usual Monte Carlo estimate corresponds to 	
\emph{uniform weighting}: $ w_{k}= 1/K, 1 \leq k \leq K$. Even for non-mixing MCMC, averaging using equal weights can outperform using any single chain \citep[e.g.,][]{hoffman2020}, but it should be possible to do better. Equal weighting is convenient, but is not in general justified, and the result can strongly depend on starting points.

The present paper provides a practical and scalable solution to the problem of representing multimodal posterior distributions through sampling when all that is available are non-mixing chains.
We propose to \emph{stack} them and compute the optimal weights in order to minimize the  prediction loss. 
Stacking \citep{wolpert1992stacked, breiman1996stacked,leblanc1996combining,  clarke2003comparing} and its Bayesian variants \citep{clyde2013bayesian,  le2017bayes, yao2018using, yao2021bayesian} are model averaging techniques for combining a discrete set of fitted models in the setting where we have data $y=(y_i)_{i=1}^n$ and models $M_1, \dots, M_K$, each having its own parameter vector $\theta_k \in \Theta_k$, likelihood, and prior. When using stacking to combine Bayesian models, we first fit each model to  obtain its posterior  distribution $p(\theta_k|y, M_k)$, and we then maximize the leave-one-out log predictive density of the combined model, $$\max_{\w}\sum_{i=1}^n \log \left(\sum_{k=1}^K w_k \int_{\Theta_k}  p (y_i|\theta_k,  M_k) p\left(\theta_k| M_k,  \{y_{i^\prime}  : {i^{\prime}\neq i}\}  \right) d\theta_k  \right),$$ where $\w$ is a simplex vector of model weights.
In this paper, we extend stacking to combine multiple chains fitting the same model. The idea is simple: We explore modes using many runs of parallel inferences and random initialization, evaluate the predictive performance of each mode using cross validation, seek the weights such that the combined-chain-inference provides the optimal posterior predictions, and plug this optimal chain weight into the weighted Monte Carlo form \eqref{eq_MC_final}.
Nevertheless, directly applying Bayesian stacking for non-mixing computations involves two challenges:
\begin{itemize}
    \item The computational challenge comes from cross validation: the exact ``cross validation of modes'' is not only expensive, but also not well defined because data split (leave-data-out) can move, merge, or create posterior modes.  We  propose an importance sampling scheme to avoid the cost and ambiguity of cross validation. 
      \item  The conceptual challenge is that our goal of minimizing prediction loss is not the same as the typical goal of multimodal sampling (to approximate the exact Bayesian posterior distribution); hence the stacked-chain inference differs from  Bayesian inference. We argue that, in the presence of posterior multimodality, predictive performance is the more relevant goal.
\end{itemize}

The rest of the paper is organized as follows.
Section \ref{sec:method} details our method and practical implementation to deal with non-mixing chains for Bayesian computation.  In Section \ref{sec_toyexample}, we provide intuition by discussing various types of posterior multimodality their relations to model misspecification in a simple example.
We review related methods in Section \ref{sec_related}.
In Section \ref{sec_theory} we show the asymptotic optimality of the proposed method via a theoretical example in which the stacked-chain inference fits data better than the exact posterior density even asymptotically. The effectiveness of stacking is best demonstrated by applying it to a series of challenging problems that represent different sorts of posterior distributions that arise in applied statistics. Therefore, Section \ref{sec:example} uses chain-stacking to address posterior multimodality and slow mixing in several challenging classes of model: latent Dirichlet allocation, Gaussian process regression,  variational inference in horseshoe regression, and Bayesian neural networks.

\section{An Approach to Inference From Non-mixed Computation: Parallel Approximation and Stacking}\label{sec:method}
	
	\subsection{Analyzing and Reweighting Simulations from Multiple Chains}
	
	We are working with the general setting of data  $y=\{y_{1},\dots, y_n\}$,  model $p(y, \theta)$, and the goal of posterior inference on $p(\theta|y)$. To start, we assume we have some existing computer program that attempts to draw samples from   $p(\theta|y)$ but might get trapped in a single mode or, more generally, a small part of the distribution.  For the present paper, all that is necessary is that the algorithm produces \emph{some} set of posterior draws, which can be obtained by generic Markov chain Monte Carlo sampling such as from Stan \citep{stan2020},
	variational inference \citep{blei2017variational}, or mode-based approximation such as Laplace's method or expectation propagation \citep{vehtari2014expectation}.

	\paragraph{Step 1: Parallel evaluation.}  We run our program $M$ times from different starting points to have a chance to explore many modes or areas of the target distribution.  We also recommend an overdispersed initialization.  Using multiple starting points is not a new idea in statistical computation, but we emphasize that our goal here is \emph{exploration}, without the expectation that the chains will mix with each other, nor that all modes and separated regions are reached.  It could, for example, make sense to run the simulation algorithm in parallel on a large number of processors in a cluster.
	
	In MCMC methods, it is often easier to achieve within-chain mixing than between-chain mixing. This is especially true for distributions with isolated modes. To monitor within-chain mixing, we use split-$\widehat{R}$  \citep{vehtari2019rank}. That is, given each individual chain, we start by discarding the simulations from the warmup or adaptation phase, then we split the saved iterations into two halves (to enable detection of nonstationarity when the first and second half of a chain are discordant). 
	For each scalar parameter $x$, we denote these two halves by  $x^{(1)}, \cdots, x^{(S/2)}$ and $x^{(1+S/2)}, \cdots, x^{(S)}$.  We compute the half-wise mean $\bar x^{(1)}= \frac{2}{S} \sum_{s=1}^{S/2} x^{(s)} $ and  $\bar x^{(2)}= \frac{2}{S} \sum_{s=1+S/2}^{S} x^{(s)}$,  and the chain-wise mean $\bar x= \frac{1}{S} \sum_{s=1}^{S} x^{(s)}$.  We then compute the between- and within-half variances,
	$$B = S \sum_{m=1}^2 (\bar x^{(m)}- \bar x)^2,~W= \frac{1}{S-2} \sum_{m=1}^2   \sum_{s=1}^{S/2} ( x^{\left(s+ S(m-1)/2\right)} -\bar x^{(m)})^2.
	$$
	We define split-$\widehat{R} =  \sqrt {\frac{S-2}{S} + \frac{2B}{SW}}$. In most simulations we experimented, it is fairly easy to have split-$\widehat{R} <1.05$ for most chains, indicating good within-chain mixing.

	\paragraph{Step 2 (optional): Clustering.} 
	We can use a between-chain mixing measure such as $\widehat{R}$ \citep{gelman1992inference, vehtari2019rank} to partition the $M$ parallel simulations into $K$ clusters, each of which  approximately captures the same part of the target distribution. Label the simulations from cluster $k$ as $(\theta_{ki}, i=1,\ldots,S_k)$, with the total number of draws being $S=\sum_{k=1}^K S_k$.  This step is optional and recommended if the number of parallel runs $M$ is large.
	
	To keep notation coherent, when the clustering step is skipped, we denote $K=M$ and $\theta_{ks}$ as the $s$-th sample in the $k$-th chain.  Throughout the paper, we use $1\leq i \leq n$ to index outcome observations, $1\leq k \leq K$ to index clusters (chains, optimization runs), and $1\leq s \leq S$ to  index posterior draws.

	\paragraph{Step 3: Reweighing non-mixing chains using stacking.}   
	
	From the previous two steps, we assume $\theta_{ks}$ come from a stationary distribution $p_k(\theta|y)$, which in general do not mix, nor do they match the exact posterior $p(\theta|y)$.
	
	
	We seek an optimal weight in \eqref{eq_MC_final} that maximizes the leave-one-out cross validation performance of the distribution formed from the weighted average of the simulation draws.  This first requires estimation of the pointwise leave-one-out (loo) log predictive density \citep[lpd,][]{gelman2014understanding} from the $k$-th cluster (chain):
	\begin{equation}\label{eq_elpd}
	\log p_k(y_i|y_{-i})= \log  \int_{\theta\in \Theta}  p(y_i | \theta )  p_k(\theta| y_{1, \ldots, i-1, i+1,  \ldots, n} )d \theta, \quad i=1, \ldots, n,~ k=1,\ldots, K.
	\end{equation}
	Second, we solve
	\begin{equation}\label{eq_stacking_objective}
	\w_{1, \ldots, K}^*=\arg\max_{\w\in \mathbbm{S}(K)} \sum_{i=1}^n \log \sum_{k=1}^K w_k p_k(y_i|y_{-i}) + \log p_{\mathrm{prior}}(\w), 
	\end{equation}
	where $\mathbbm{S}(K)$ is the space of $K$-dimensional simplex 
	\begin{equation*}
		\mathbbm{S}(K)=\{w: 0 \leq w_{k} \leq 1, \forall 1\leq k\leq K; ~ \sum_{k=1}^K w_k =1 \},
	\end{equation*}
	and  $p_{\mathrm{prior}}(\w)$ is prior regularization.
	
	In Section \ref{sec_detail_implement},  we explain how to approximate  the $\log p_k(y_i|y_{-i})$ terms by importance sampling---it suffices to fit all the full data once in each chain. We will also discuss the choice of priors $p_{\mathrm{prior}}(\w)$.

	Finally,  
	plugging the stacking wights $w_1^*,\ldots ,w_K^*$ into \eqref{eq_MC_final} yields the chain-weighted Monte Carlo estimates. The resulting approximation of the target distribution uses $\sum_{k=1}^K {S_k}$ draws,  with each $\theta_{ks}$ having weight $w_k^*/S_k$.
	
	
	\paragraph{Step 4: Monitoring convergence.}
	After $K$ parallel runs, we cannot exclude the possibility that another local mode or separated posterior region has been overlooked.  
	When there is a discrete combinatorial explosion, it is essentially impossible to capture the full support of the distribution.  So we are implicitly assuming that we have a rough sense of the support of most of the posterior mass, or, conversely, that we were previously willing to approximate the target distribution using a single mode, in which case we would hope a multimodal average to be an improvement. 
	
	On the other hand, there is no need to capture all modes that are predictively identical. We monitor the weighted log predictive density as a function of how many components are added in stacking. Ideally, we should test it over  an independent hold-out test data set,
	and stop when the log predictive density  of the stacked posterior reaches the maximum. Alternatively, we can use cross validation.  For each $K^\prime \leq K$, obtain stacking weights $w^{K^\prime}_k$ from chains $1, \dots, K^\prime$, and monitor the stacked log predictive density as a function of the number of chains $K^{\prime}$, which typically monotonically increases:
	\begin{equation}\label{eq_monitor}
	\mathrm{lpd}_{\mathrm{loo}} (K^\prime) =  \sum_{i=1}^n \log \sum_{k=1}^{K^\prime} w^{K^\prime}_k p_k(y_i|y_{-i}),  ~1\leq K^\prime \leq K.
	\end{equation}
	We terminate if $\mathrm{lpd}_{\mathrm{loo}} (K^\prime)$ becomes relatively stable. Otherwise, we sample extra chains  and repeat steps 1--4 on all chains.

	\subsection{Practical Implementation}\label{sec_detail_implement}
	\paragraph{Leave-one-out posterior distributions.}
	Let $p_k(\theta|y)$ be the stationary distribution from which the $k$-th cluster (chain) is drawn. Working with the exact leave-one-out distributions $p_k(\theta|y_{-i})$ in $\eqref{eq_elpd}$
	is not only computationally intensive (requiring the model to be fit $n$ times) but also conceptually ambiguous:  Using full data and given initialization, the sampler obtains $\theta_{k1}, \dots, \theta_{kS_k}$ from the $k$-th region.  After $y_i$ is removed,  what if the sampler from the same initialization reaches  another mode, or what if there is a phase transition and there are no longer $K$ clusters?

	We avoid the ambiguity  by defining  $p_k(\theta|y_{-i})$ to be
	\begin{equation}\label{eq_loo}
	p_k(\theta|y_{-i}) \coloneqq  \frac{p_k(\theta|y)   / p(y_i|\theta)} {\int_{\theta\in \Theta} p_k(\theta|y)   / p(y_i|\theta)} .
	\end{equation}  	This definition is backward compatible with the usual cross validation of models, in which the leave-one-out posterior density (of a model) is
	$
	p(\theta|y_{-i}) \propto  p(\theta|y_{-i}) p(\theta) =  p(\theta|y) / p(y_i | \theta).
	$ 

	\paragraph{Efficient approximation of leave-one-out distributions.}
	We use Pareto smoothed importance sampling \citep[PSIS,][]{  vehtari2015pareto} to compute the defined LOO posterior \eqref{eq_loo}.  It suffices to only fit the full model once per chain.  For each chain $k$, we obtain the raw leave-one-out importance ratios $1/p(y_{i} |\theta_{ks} ), \, i=1,\dots,n$ and stabilize these by replacing the largest ratios by the expected order statistics in a fitted generalized Pareto distribution and followed by right truncation.  Labeling the Pareto-smoothed  importance ratio as ${r_{iks} }$, we approximate $p_k(y_i|y_{-i})$ by 
	\begin{equation}\label{eq_LOO_approx}
	p_k(y_i|y_{-i})  \approx  \frac{ \sum_{s=1}^{S_k} p_k(y_i|\theta_{ks}) r_{iks} }{ \sum_{s=1}^{S_k}  {r_{iks} }}, ~ k=1, \dots, K, ~ i=1, \dots, n .
	\end{equation}
	This is asymptotically ($S_k\to \infty$) unbiased and consistent to the definition \eqref{eq_loo}.  The finite-sample reliability and convergence rate can be assessed using the estimated shape parameter $\hat k$ of the fitted generalized Pareto distribution.  We refer to \citet{vehtari2017practical, Vehtari2019loo} and Appendix B of this paper for detailed algorithms and software implementation.
	
	In summary, after parallel sampling, the extra computation costs of stacking only involve summations in \eqref{eq_LOO_approx} and a length-$K$-vector optimization in \eqref{eq_stacking_objective}, which are negligible compared with the cost of sampling.
	
	\paragraph{Prior on stacking weights.}
	Extra priors beyond a simplex constraint in model averaging have been considered \citep{le2017bayes, yao2018using} but seldom applied in practice. 
	Under a flat prior $p_{\mathrm{prior}}(\w)= 1$, the optimum in \eqref{eq_stacking_objective} is nonidentified and numerically unstable if two simplexes  $w^\prime\neq w^{\prime \prime}$ entail the identical prediction $\sum_k w^\prime_k p_k(\cdot|y)= \sum_k w^{\prime \prime}_k p_k(\cdot|y).$   We need an informative prior for the \emph{predictive power} versus \emph{Monte Carlo error} tradeoff. 
	
	If all chains are distributed identically, and within chain sampling is independent, the variance of \eqref{eq_MC_final} will be  
	$\Var  \left(    \sum_{k=1}^K \sum_{s=1}^{S_k} {w_k}{S_k}^{-1}  h(\theta_{ks})  \right) =      \sum_{k=1}^K    {w_k}^{2}{S_k}^{-1}    \Var \left(h(\theta)  \right).$ As a function of simplex $\w$, this variance is minimized when  $w_k= S_k/\sum_{k^\prime} S_{k^ \prime}$. This justifies the uniform weights $1/K$ in the usual multi-chain Monte Carlo scheme where, after complete mixing, any weighting yields unbiased estimates.
	
	Further, when the $k$-th chain has an effective sample size $S_{\mathrm{eff}, k}$ \citep{vehtari2019rank}, we approximate the variance of the Monte Carlo estimate \eqref{eq_MC_final} to be
	$\Var  \left(    \sum_{k=1}^K \sum_{s=1}^{S_k} {w_k}{S_k}^{-1}  h(\theta_{ks})  \right) =      \sum_{k=1}^K    {w_k}^{2}{S_{\mathrm{eff}, k}}^{-1}    \Var \left(h(\theta)  \right), $
	whose minimum will be attained at $w_k= S_{\mathrm{eff}, k}/\sum_{k^\prime}S_{\mathrm{eff}, k^\prime}$. This also suggests we can estimate the  the effective sample size of $\w$-weighted samples by: 
	$$\hat S_{\mathrm{eff}}:= \left( \sum_{k=1}^K {w_k}^{2}{S_{\mathrm{eff}, k}}^{-1}\right) ^{-1}.$$
	
	To reduce Monte Carlo error, we partially pool stacking weights using a Dirichlet prior with a tuning scale parameter $\lambda>0$ that controls the amount of partial pooling, 
 \begin{equation}\label{eq_prior}
p_\mathrm{{prior}} (w_{1, \dots, K})  = \mathrm{Dirichlet} \left(   \frac{\lambda  S_{\mathrm{eff}, 1}}  {  \sum_{k^\prime=1}^K   S_{\mathrm{eff}, {k^\prime}   }}, \dots,   \frac{\lambda  S_{\mathrm{eff}, K}} { \sum_{k^\prime=1}^K   S_{\mathrm{eff}, k^{\prime}}} \right).
\end{equation}
	We add this regularization term into \eqref{eq_stacking_objective}.  If  $\lambda=1$ and $S_{\mathrm{eff}, k}$ is equal for all $k$, \eqref{eq_stacking_objective} becomes the unregularized  Bayesian stacking. 
	For any $\lambda>1$, the optimization is strictly convex on $\w$. 
	If $\lambda \to \infty$ and $S_{\mathrm{eff}, k} \propto S_k$, \eqref{eq_stacking_objective} results in the usual Monte Carlo estimate $w_{k}/S_{k}=1/S$. 	
	Ideally, $\lambda$ can be further tuned using hold-out data or extra cross validation. In later experiments of this paper, we simply use $\lambda=1.001$ as a rule of thumb.

	\paragraph{Thinning and importance resampling.}
	For settings where it is inconvenient to work with weighted simulation draws, we can perform thinning to obtain a set of $S_{\mathrm{thin}}$ simulation draws approximating the weighted mixture of $K$ distributions. We further adopt quasi Monte Carlo strategy to reduce variance.  Given weights $\{w_k\}_{k=1}^K$ for $K$ clustered simulation draws $\{\theta_{ks}\}_{k=1, s=1}^{K,~~ S_k}$, and an integer $S_{\mathrm{thin}} \leq \inf_k (S_k / w_k)$,  we first draw a fixed-sized $S_k^*=  \floor{S_{\mathrm{thin}} w_k}$ sample randomly without replacement from the $k$-th cluster, and then sample the remaining $S_{\mathrm{thin}}- \sum_{k=1}^K S_k^*$ without replacement with the probability proportional to  $  \left( w_k- S_k^* /S_{\mathrm{thin}}\right)$ from  cluster $k$.  
	
We  implement  related functions together to facilitate chain-stacking in an \texttt{R} package \texttt{loo} that works seamlessly with Stan. See Appendix B for an example.

\section{Modes:  The Good, the Bad, and the Ugly}\label{sec_toyexample}
Before more theoretical discussions,  we first develop intuition by considering variations on a theme:  four examples of mixture models that demonstrate  different types of posterior multimodality: where do the modes come from and when do they matter. 
	\begin{figure}[!ht] 
		\centering
		\vspace{0em}
		\includegraphics[width=1\linewidth]{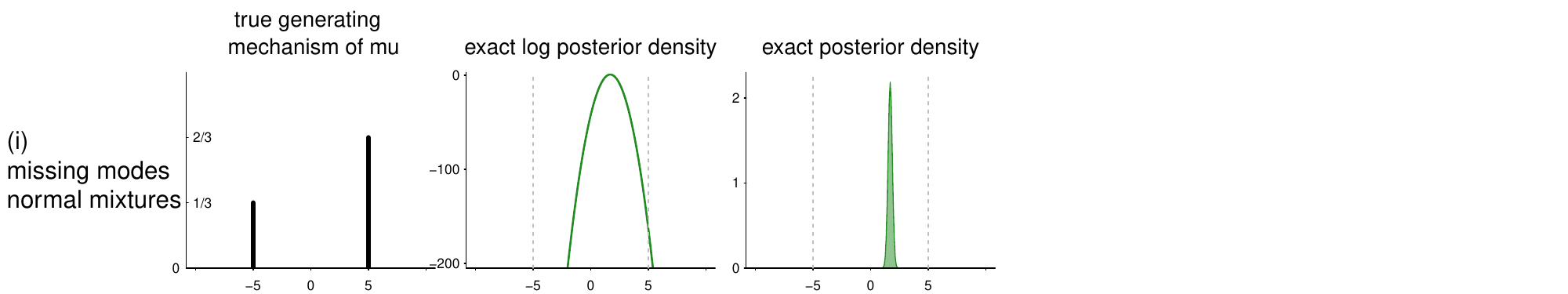}
		\includegraphics[width=1\linewidth]{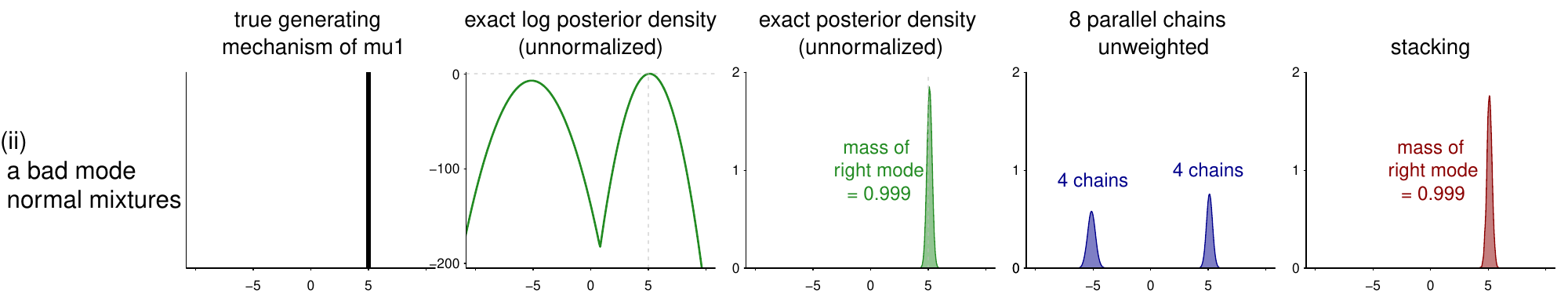}
		\includegraphics[width=1\linewidth]{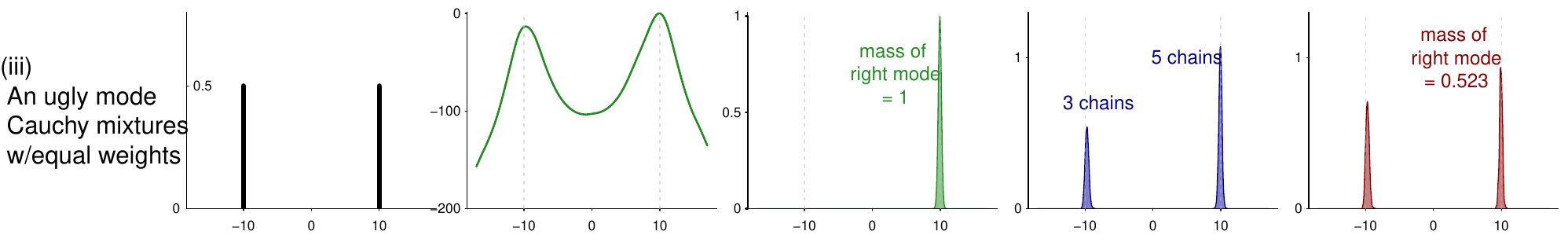}
		\includegraphics[width=1\linewidth]{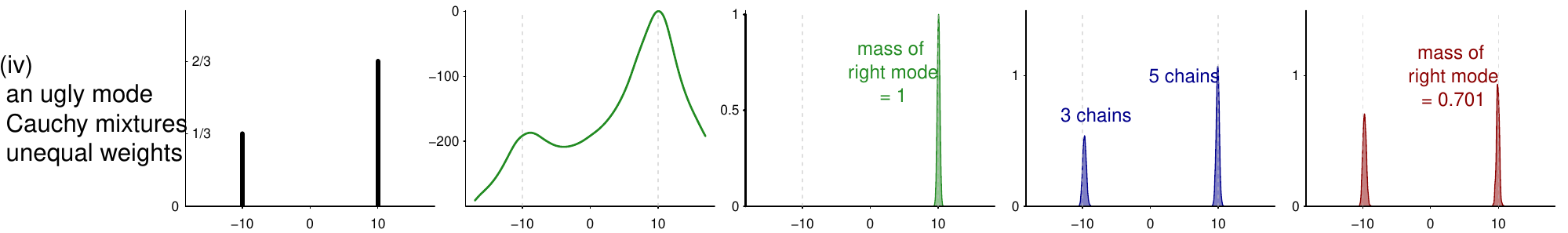}
		{\footnotesize
			\centering	
			\vspace{3em}
			\begin{tabularx}{ \textwidth}{s s s  b b b }
				{\color{black} {\bf Summary:}   example} &  {\color{black} \#   modes \newline in  DG} &{\color{forestgreen} \#   modes in  posterior }&{\color{forestgreen}posterior $\to$    DG as $n\to \infty$?}  & {\color{Darkblue}unweighted chains  approx.  DG?}& {\color{darkred}stacking  approx. DG?}\\ 
				\hline 
				i & 2 & 1 & \xmark   & \xmark & \xmark \\ 
				ii & 1 & 2 & \cmark  & \xmark & \cmark \\ 
				iii & 2 & 2 & \xmark  & \xmark & \cmark \\ 
				iv  & 2 & 2 & \xmark  & \xmark & \cmark 
			\end{tabularx} 
		}	
		\vspace{-3em}
		 \caption{\em Under a multimodal data generating mechanism, the exact Bayesian posterior can miss  the modes (row (i)) or over-concentrate at one mode (rows (iii)--(iv)).    
			Stacking, our proposed method, approximates the data generating process well in (ii)--(iv).  
			The sample size is $n=30$ in (i)--(ii) and $n=100$ in (iii)--(iv). 
		}
		\label{fig_mode1}
	\end{figure}

\begin{enumerate}[label=(\roman{enumi}),  leftmargin=1.5em]
	\item A missing mode:
	We draw $n$ points $y=(y_i,\dots,y_n)$ independently from the mixture, $\frac{2}{3}\mbox{normal}(5,1) +\frac {1}{3}\mbox{normal}(-5,1)$.  We  
	fit the model $y_i | \mu \sim \mbox{iid~} \mbox{normal}(\mu,1)$ with a flat prior on $\mu$. The true data generating process (DG) is expressed by $\mu \sim \frac{2}{3} \delta(5) +  \frac{1}{3} \delta(-5)$, that is, a mixture of point masses at $\mu=5$ and $\mu=-5$ with mixing probabilities 2/3 and 1/3. 
	But the Bayesian posterior density, $p(\mu|y) = \mbox{normal}(\bar y, 1/\sqrt{n})$, is unimodally concentrated at $\mu=\bar y \approx 5/3$ and cannot catch the two modes in data.
	
	\item A bad mode: 
	With the same data $y$ above, now we fit a two-component normal model $y\sim \frac{2}{3}\mbox{normal}(\mu_1,1)+ \frac{1}{3}\mbox{normal}(\mu_2,1)$ with known mixture probability and a flat prior on $\mu_1, \mu_2$. The model is identifiable, but the resulting posterior is bimodal, centered around $(\mu_1, \mu_2) = (5, -5)$ and $(-5, 5)$ respectively. Asymptotically ($n\to \infty$) the posterior converges to the first mode, thereby the data generating process, but the existence of a second artifact mode both challenges the sampling and compromises the prediction with finite data sample size.  In Figure \ref{fig_mode1} we simulate $n=30$ data points and run eight parallel chains. Four chains are trapped in the ``wrong'' mode. 
	
	\item  An ugly mode:
	We generate data $y_1, \dots, y_{n}$ iid from $\frac{1}{2}\mathrm{Cauchy}(10,1) + \frac{1}{2}\mathrm{Cauchy}(-10,1)$. We fit a one-component model $y\sim \mathrm{Cauchy}(\mu,1)$ with a flat prior.  The true data generating process is expressed by $\mu \sim \frac{1}{2} \delta(10) +  \frac{1}{2} \delta(-10)$. In the limit ($n\to \infty$), the posterior density will be  concentrated at one of two points $\mu \approx \pm 9.8$.  In the simulation with $n=100$, the right-side posterior mode contains almost 100\% mass (up to the precision $10^{-6}$). The induced predictive model then only describes half of the data.  Stacking, as implemented in this paper, assigns a weight of $0.52$ to the right-side mode, achieving a much better prediction compared to the data generating process.
	
	\item Another ugly mode:
	We draw $n$ data points $y_i,\dots,y_n$ independently from the mixture model, $\frac{2}{3}\mathrm{Cauchy}(10,1) + \frac{1}{3}\mathrm{Cauchy}(-10,1)$ and again fit a one-component model $y\sim \mathrm{Cauchy}(\mu,1)$ with a flat prior.  The posterior $p(\theta|y)$  carries almost all masses on the right-side mode $\theta \approx 10$, while our proposed method still approximates the true data generating process. 
	\end {enumerate}

Figure \ref{fig_mode1} illustrates the true distributions of  $\mu$, the unnormalized log Bayesian posterior density $\log p(\mu|y)$, the  unnormalized Bayesian posterior density $p(\mu|y)$,  the distribution of uniformly weighted chains (aggregate 8 parallel chains without adjustment), and the stacked-chain inference of $\mu$. Each row is one example above. 
In example (ii), one of the modes is purely an artifact: Not drawing a posterior sample around it improves finite-sample predictions.  Such artifact-type modes are found in cases of prior-data conflict, label-switching, aliasing \citep{bafumi2005practical}, mixture and cluster-based models \citep{stephens2000dealing, blei2003latent},   and hierarchical models \citep{liu2003posterior}. But in other examples, the data generating process (DG) can be expressed via a bimodal distribution on $\mu$. In example (i), the Bayesian posterior $p(\mu|y)$  converges to some middle point. In (iii)--(iv), the posterior overconfidently concentrates at one of the modes and ignores the other, even when data are truly generated from these two modal points with equal probability (iii).

Our proposed approach of weighting modes using stacking is suitable for all these scenarios (except example (i) as there is only one mode).   We will revisit this Cauchy mixture in Section \ref{sec_theory} and prove its limiting behavior analytically, in which Bayesian inference almost surely overconfidently concentrates, while our proposed method recovers the \emph{true} data generating process from the \emph{wrong} model and \emph{wrong} inference.
	
	

	\section{Related Work}\label{sec_related}
Our work is mostly motivated by model averaging methods.
Under an ideal assumption that all regions of the posterior distributions have been fully explored, chain/cluster $k$ samples from   a local distribution $p_k(\theta | y)$ on an  attraction regions $\Theta_{k}$,  and these regions are well separated; that is,	\begin{equation}\label{eq_importance_weight_assumption}
	\Theta=\bigcup_{k=1}^K \Theta_k;    \quad  \forall  k^{\prime}\neq k, ~p_k \left( \Theta_{k^\prime}\right) \approx 0;   \quad \mathrm{s.t.} ~ \forall \theta \in \Theta_k, ~   p(\theta|y) \approx  \alpha_k p_k (\theta|y ). 
	\end{equation}
 Using weights $w_k= \alpha_k$ in the weighted Monte Carlo expression matches the usual Monte Carlo computation from the exact posterior draws. This $\alpha_k$ weighting can be interpreted as   Bayesian model averaging \citep[BMA; ][]{madigan1996bayesian, hoeting1999bayesian} on a discrete model space where model $k$ has posterior density $p_{k}(\theta|y)$. The marginal  likelihood of cluster $k$ is thus
 $\int_{\Theta} p(y, \theta) \mathbbm{1}(\Theta_k) d\theta  \approx  \alpha_k$. To evaluate this integral has the usual difficulty as in marginal likelihood computation.

As a computationally-easier alternative to BMA,  \citet{yao2018using}  introduced pseudo-BMA weighting for model averaging. Applying to our context,  the  pseudo-BMA  weight for cluster $k$ is
	$$\alpha_k^{\mathrm{(pseudo\!-\!BMA)}}  \propto   \exp\left(  \sum_{i=1}^n \log p(y_{i} |y_{-i} ,\Theta_k) \right) 
	\approx  \exp\left( \sum_{i=1}^n \log \sum_{s=1}^{S_k}  \frac{r_{iks}   p(y_{i} |\theta_{ks})}{r_{iks}}   \right) ,
	$$
	where $r_{iks}$ is the same  leave-one-out importance ratio  in \eqref{eq_LOO_approx}. To stabilizes the weights,  
	\citet{yao2018using}  further  recommended the Bayesian bootstrap. \citet{fong2019scalable} adopted a similar strategy to tackle multimodal sampling. 
	
	In comparison, BMA is fully Bayesian under assumption \eqref{eq_importance_weight_assumption} and the correct model specification.  However, in many approximate inferences, the local approximation  $p_k(\cdot|y)$ is underdispersed and BMA loses mass. When using multi-chain MCMC, $\Theta_k$ are often duplicate (without clustering) or overlapped, making BMA weighting sensitive to the distribution of starting points of chains. Furthermore, \citet{yao2018using} noted that  BMA and pseudo-BMA can overweight ``bad'' modes when they are oversampled.  This is related to the discussion by  \cite{geyer1992practical} that a simple unweighted average over non-mixing chains only helps when the starting distribution is close to the target density---the scenario in which other naive methods will work, too. In contrast, our proposed stacking approach is invariant to chain duplication and not sensitive to chain initialization other than the requirement that all relevant modes are explored by random starting points. This is	because the optimization \eqref{eq_stacking_objective}  only depends on the set of distinct densities $p_k(\theta|y)$, not the proportion of how many chains are trapped to these densities. 

 Our approach uses a divide-and-conquer strategy that is embarrassingly parallelizable and eliminates between-chain communication, which often dominates the budget of parallel computations \citep{scott2016bayes}.  Because of the fast mixing  rate of  Hamiltonian Monte Carlo (HMC) in log-concave distributions \citep{beskos2013optimal}, the bottleneck of modern Bayesian computation is often not the input dimension, but the slow mixing rate arising from awkward geometry of metastable distributions. In general, Bayesian inference can be more scalable in the advent of  parallel distributed computation.  Various subsampling methods have been introduced that distribute data batches to parallel nodes and aggregate the resulting inference \citep{huang2005sampling, welling2011bayesian,angelino2016patterns,    mesquita2019embarrassingly, quiroz2019speeding}.   These methods typically rely on approximations to rescale the subsampled posteriors, and can work poorly with posterior multimodality.  
	
It is not a new idea to use random starting points.
\citet{gelman1992inference} used multiple sequences and importance resampling to approximate the posterior distribution, where each individual chain was iteratively constructed from a local Student-$t$ approximation at posterior mode. However, a poor initial point can still lead to slow convergence \citep{geyer1992practical} because of the use of importance sampling.  In our approach, we are less concerned about starting points and only prefer them to be overdispersed.   \citet{raftery1992many, raftery1992} suggested abandoning poor initial points coming with slow convergence rate and high autocorrelation by restarting.  In the context of multimodality, it is hard to tell if this represents a poor initialization (that sits near the boundary of an attraction region) or a bad mode. A restart may lose the chance to explore some posterior regions.
	
	Our convergence criteria in Section  \ref{sec_detail_implement} are similar to the early approaches on stochastic optimization stopping rules  following the capture-recapture model \citep{good1953population, robbins1968estimating,  finch1989probabilistic}. Those analyses were focused on the convergence in parameter space, while ours are directly targeted at the outcome space and are thereby more applicable to models with a large number of disjoint but functionally identical modes.

The stacking strategy is applicable  to multiple runs of approximate inference; see examples in Section \ref{sec_example_vi}. Using mixture distributions to enrich the expressiveness of variational Bayes is not new. Earlier works have used a mixture of  mean-field approximations to  match the posterior   \citep{bishop1998approximating,jaakkola1998improving,  gershman2012nonparametric,  ranganath2016hierarchical,   gal2016dropout, miller2017variational,  chang2019ensemble}. However, a direct application of mixture variational methods can be prohibitively expensive in large models, and weights are often fixed to ease the cost.
	Stacking does not need to specify either the parametric form of the mixture or the number of mixture components, both of which adapt to data and prevent extra model misspecification.

Lastly, there is rich literature on MCMC techniques that attempt to sample from  a multimodal density. In some cases it is possible to collapse multimodality using reparameterization \citep{papaspiliopoulos2007general, johnson2012variable, betancourt2015hamiltonian, gorinova2019automatic}, but this cannot be automated for general problems. Several schemes have been proposed for sampling from  distributions with isolated modes by adding an auxiliary temperature parameter to enhance the transition probability between modes; these methods include annealing \citep{kirkpatrick1983optimization, bertsimas1993simulated}, parallel tempering \citep{hansmann1997parallel, earl2005parallel}, simulated tempering \citep{marinari1992simulated, neal1993probabilistic}, auxiliary variables \citep{wang2001efficient},  and path sampling \citep{gelman1998simulating, yao2020adaptive}.  These methods involve sampling from a tempered distribution $p(\theta | \lambda) \propto p(\theta|y )^\lambda$, with the idea that this distribution is flatter conditional at a smaller $\lambda$ and easy to sample from.
Although tempering-based methods are popular in statistical physics and molecular dynamics problems, they are sensitive to implementation and tuning, which can make them less appropriate for general statistical computation, and their theoretical mixing rates drop quickly in high dimensions \citep{bhatnagar2004torpid}.
Moreover, the metastability of sampling  comes from both the energetic (two modes are distinct) and entropic (two regions are connected through a narrow neck) barriers. Increasing the temperature does not ease the entropic barrier, which is a common problem with hierarchical models.
For all these reasons, tempering methods appear unlikely to work in large statistical models with the scale that we consider in the present paper. 	
To confirm, \citet{yao2020adaptive} reported the failure of simulated tempering 
when applied to multimodal posterior distribution from our latent Dirichlet allocation example (Section \ref{sec_lda}).

	\section{Asymptotic Analysis in a Theoretical Example}\label{sec_theory}
	In this section, we analyze the asymptotic behavior of  stacked-chain inference.   We first show the overconfidence of Bayesian inference in the existence of posterior multimodality, 
	while in contrast, the proposed method is no worse than chain-picking. Then we derive a closed-form solution in a theoretical example to show that with model-misspecification and multimodal posterior,
	chain-stacking can be predictively better than the  exact inference.     Proofs and related lemmas are in Appendix A.

		\subsection{Overconfidence of Bayesian Inference}\label{sec_multimodal_dg}
	Given data $y_{1},\dots, y_n$ generated independently and identically distributed from an unknown data generating process: $p_{\mathrm{true}}(y)$, and a potentially misspecified model $y| \theta\sim f(y|\theta)$ and prior $p(\theta),  \theta\in \Theta$, when the  sample size $n$ goes to infinity and regularization conditions apply, the limiting Bayesian posterior density will be almost surely supported on the set of global modes \citep[][]{berk1966limiting}:
	$
	\mathcal{A}=\Bigl\{\theta^* \in \Theta: \E_{\tilde y \sim p_{\mathrm{true}}}  \log f( \tilde y |\theta^*) = \max_{\theta\in \Theta} \E_{\tilde y \sim p_{\mathrm{true}}}  \log f( \tilde y |\theta)     \Bigr\}.
	$

		\begin{figure}
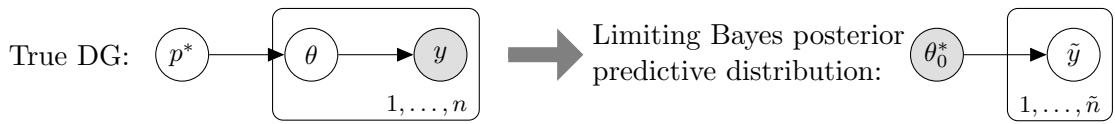

		\centering
		\tikz{
			\node[obs] (y) {$y$};%
			\node[latent,left=of y,xshift=0cm,fill] (s) {$\theta$}; %
			\node[latent,left=of s,xshift=0cm,fill] (ss) {$p^{*}$}; %
			\plate [inner sep=.15cm,yshift=.1cm] {plate1} {(y)(s)} {$1, \dots, n$}; %
			\node[left of=ss, xshift=-0.5cm]   { True DG:};
			
			\node[latent, right=of y,xshift=6.7cm,fill] (y2) {$\tilde y$};%
			\node[obs,left=of y2,xshift=-0.1cm,fill] (s2) {$ \theta_0^*$}; %
			\plate [inner sep=.15cm,yshift=.1cm] {plate1} {(y2)} {$1, \dots, \tilde n$}; %
			\node[left of=s2, xshift=-1.5cm](l)   {  
				\begin{tabular}{l}
					Limiting Bayes posterior  \\
					predictive distribution: \\
				\end{tabular}
			};
			\draw[
			-triangle 90,
			line width=0.9mm,
			gray,
			postaction={draw, line width=2mm, shorten >=2mm, gray, -}
			] (0.9,0) -- (1.9,0);
			\edge {ss} {s}
			\edge {s} {y}
			\edge {s2} {y2}
			
		}\caption{\em When the parameter $\theta$ is randomly drawn from a distribution $p^*$ in the data generating process \eqref{eq_what_is_true_model}, the limiting posterior inference $p(\theta|y)$ almost surely converges to a point estimate $\theta_0^*$.}\label{fig_graph}
	\end{figure}

	Such limiting behavior restricts the expressiveness of posterior predictions. When data are generated from one parameter $\theta_0$ in the model (an $\mathcal{M}$-closed view), $p_{\mathrm{true}}=f( \cdot | \theta_0)$, the posterior will be asymptotically concentrated at  $\theta_0$. But otherwise, the limiting predictive distribution seeks the closest distribution to data generating process in terms of Kullback–Leibler (KL) divergence, as we can rewrite the set  $\mathcal{A}$ as 
	\begin{equation}\label{eq_limiting}
	\mathcal{A}=  \arg\min_{\theta\in \Theta} \mathrm{KL}\Bigl(p_{\mathrm{true}}(\cdot) ~ ||  ~ f( \cdot |\theta)   \Bigr),   \quad \forall \eta>0,~ \mathrm{Pr}_{\mathrm{Bayes}}( ||\theta -  \mathcal{A}||<\eta   \mid y_{1 \dots, n})  \xrightarrow{n\to\infty, ~a.s.}1,
	\end{equation}
	The asymptotic predictive distribution is from some \emph{point} estimate $\theta^*\in  \mathcal{A}$. But ideally we would fully use the expressiveness of the model and find the optimal \emph{probabilistic} inference $p_\mathrm{optimal}(\theta)$ from some space $\mathcal{F}$ that renders the best prediction for future unseen data
	\begin{equation}\label{eq_limiting_goal}
	p_\mathrm{optimal} = \arg\min_{\tilde p \in \mathcal{F}}  \mathrm{KL}\Bigl(p_{\mathrm{true}}(\cdot) ~ ||   ~   \int_{\Theta}  f( \cdot |\theta)  \tilde  p (\theta) d \theta \Bigr).
	\end{equation}
	In particular, if the model is expressive enough such that  there is a density $p^*(\cdot) \in \mathcal{F}$ generates the data by 
	\begin{equation}\label{eq_what_is_true_model}
	p_{\mathrm{true}}(\tilde{y})= \int_{\Theta} f(\tilde y| \theta) p^*(\theta) d\theta,
	\end{equation}
	then this $p^*(\cdot)$ is  one solution to  \eqref{eq_limiting_goal} because the log score is proper.

When the posterior distribution is multimodal,   even though the multimodality suggests that the true data are unlikely generated to have been from any single  parameter in the model,   the Bayesian posterior still concentrates to one of the modes in the limit, so  that the density family $\mathcal{F}$ is a set of  Dirac delta functions: $\mathcal{F}=\{\delta(\theta_0) \mid  \theta_0 \in \Theta  \}.$  In contrast, stacking solves \eqref{eq_limiting_goal}  with a bigger density space, 	$
	\mathcal{F}= \left\{ \sum_{k=1}^Kw_k p_k(\theta|y):  \w\in \mathbbm{S}(K)  \right\},
	$ that is constructed from sampled  clusters, whose  solution  generally  not concentrates on a point.


	\subsection {Optimality of the Stacked Predictive Distribution}\label{sec_theory_1}
    The stacking weights are {\em not} the same as  posterior masses of each mode.  Even asymptotically, minimizing cross validation errors is different from integrating the target distribution.  Corollary \ref{them_op} affirms that the stacked inference is optimal---it asymptotically maximizes the expected log  predictive densities (elpd) among all linearly weighted combinations of  parallel chains of form \eqref{eq_MC_final}. This corollary is a consequence of Theorem 2.4 of \cite{le2017bayes}, with the difference that we have redefined cross validation via  \eqref{eq_loo}.
	
	\begin{corollary} \label{them_op} Assuming we draw $S$ posterior samples in each chain from their stationary distribution $p_k$, and we approximate the leave-one-out distribution by PSIS as in \eqref{eq_LOO_approx}, $p^S_{k, -i}(y_i)=  \sum_{s=1}^{S_k} p_k(y_i|\theta_{ks}) r_{iks}  / \sum_{s=1}^{S_k}  {r_{iks} },$ then for a fixed number of chains $K$ and a fixed weight vector $\w$, when in the limit of both the size of observations $n$ and  number of posterior draws $S$,  under regularities conditions (see Appendix), the objective function in stacking  converges to stacked elpd:
		$$ \frac{1}{n}\sum_{i=1}^n \log \sum_{k=1}^K w_k p^S_{k, -i}(y_i)  -  \mathrm{E}_{\tilde y| y_{1:n} }  \log \sum_{k=1}^K w_k p_k(\tilde y| y_{1:n} )    \xlongrightarrow{\text{$L_2$}}   0, \quad n\to \infty, ~S\to \infty .$$
	\end{corollary}
	
	
	
	
	
	
	
	
	\subsection{Cauchy Example Revisited: When Can Stacking Outperform Exact Bayes?}\label{sec_DG_converge}
	Let's revisit the Cauchy mixture example in Section \ref{sec_toyexample} and Figure \ref{fig_mode1}. Consider univariate observations $y_{1,\dots n}$  iid from the data generating process,
	$$\mathrm{DG}: ~y_i \sim  \mathrm{Cauchy}\left( (2 z_i -1) a, 1\right) ,   \quad  z_i\sim  \mathrm{Bernoulli} (p_0), \quad  i=1, 2, \dots, n.$$
	In other words, $y$ is either Cauchy$(a, 1)$ or  Cauchy$(-a, 1)$ with probabilities $p_0$ and $1-p_0$,  where the location $a>0$ and probability $p_0 \in[0.5, 1]$ are unknown constants (the $0\leq p_0 < 0.5$  counterpart is symmetric and hence omitted). We denote the density of this data generating process by $p_{\mathrm{true}}(y)$.
	
	We now fit $y$ with the iid  Cauchy likelihood with unknown parameter $\mu$ and a prior $p_0(\mu)$ that has full support on $\R$,
	$$\mathrm{Model}: ~ y_i \sim  \mathrm{Cauchy}(\mu, 1), \quad  \mu \sim p_0(\mu), \quad \mu\in \R.$$
The data generating process can be expressed from this model if an inference $\theta$ is given by  a mixture of two points, $$\mathrm{express~DG~in~Model}: ~ \mu \sim p_0 \delta(a)+ (1-p_0) \delta(-a).$$ 

	The following theorems characterize the behavior of  modes and the concentration of exact full Bayesian inference in the limit of large $n$. 
	
	\begin{theorem}\label{thm_post_mode} There exists a deterministic function $\xi(a)$   (see Lemma \ref{lemma_h}), with $\xi(2) =0.5$ and $\xi(\infty) =1$, such that the  modality of posterior density $p(\mu | y_1, \dots, y_n)$ has a closed form determination: 
		
		(a) For any $a>2$, and  $p_0\geq  \xi(a)$, there exists a large $N$, such that for all $n>N$, the posterior is unimodal.  The peak is near $\mu=a$ for a large $a$.
		
		(b) For any $a>2$, and   $0.5 \leq p_0 <  \xi(a)$, there exists a large $N$, such that for all $n>N$, the posterior is bimodal. The two local maximums  are near  $(-a, a)$ for  a large $a$.
		
		(c) For any $0<a<2$, there exists a large $N$, such that for all $n>N$, the posterior is unimodal with global maximum between 0 and $a$.   If further $p_0=0.5$, the maximum is at 0.
		
		(e) When $a>2, p_0=0.5$ and equipped a symmetric prior $p(\mu)=p(-\mu)$, there exists a large $N$ such that, for all $n>N$, the posterior is always bimodal with two maximums, which  asymptotically ($n\to \infty$) converge to $\mu=\pm \sqrt{a^2-4}$. 
	\end{theorem}
	
	\begin{figure}
		\centering
		\begin{subfigure}{0.48\textwidth}
			\centering
			\includegraphics[height= 4.9cm]{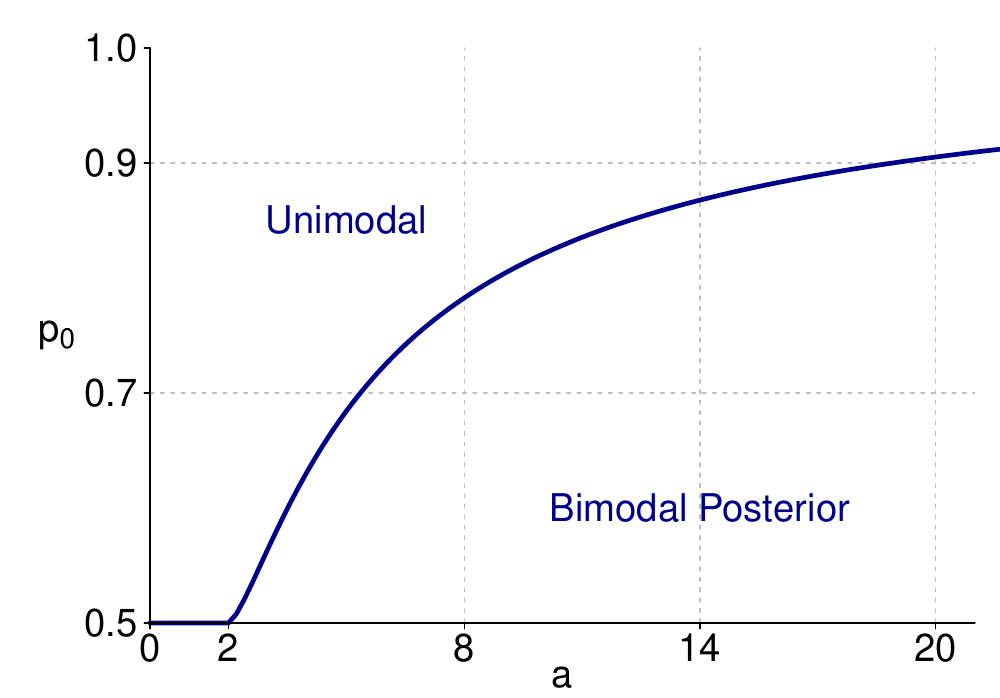}
		\end{subfigure}	
		~ 
		\begin{subfigure}{0.49\textwidth}
			\centering
			\includegraphics[height= 4.9cm]{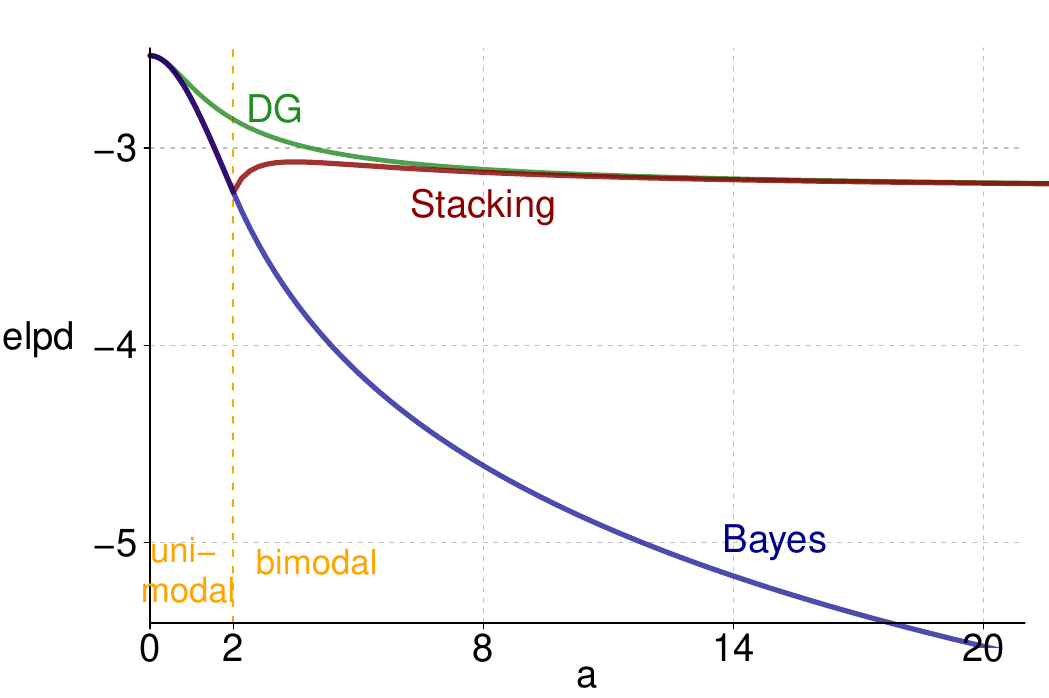}
		\end{subfigure}	
		 	\caption{\em  Left: the deterministic function $\xi(a)$.  For any $a>2$ the posterior is  bimodal with a large $n$ if and only if  $p_0 <\xi(a)$. Right: the elpd of the true data generating process and the asymptotic ($n \to \infty$) elpd  of full Bayes and multi-chain stacking at $p_0=0.5$. When $p_0=0.5$, $a<2$ the posterior is unimodally spiked at 0, and stacking is identical to Bayes.}\label{fig_xi_function}
	\end{figure}

	\begin{theorem}\label{thm_post_converge}
		(a) For any $a>2$, and $p_0>0.5$, the  posterior distribution $p(\theta|y_1, \dots, y_n)$  converges  in distribution to  a point mass $\delta (\gamma)$ as $n\to \infty$, where $\gamma = \gamma(p_0, a)$  depends on $p_0$ and $a$.
		
		(b) For any $a>2$, $p_0=0.5$, a prior that is symmetric $p(\mu)=p(-\mu)$, the posterior distribution $p(\theta|y_1, \dots, y_n)$ is asymptotically only charged at two points  $\pm  \gamma$, with a closed form expression   $\gamma =\sqrt{a^2-4}$.   More precisely,  the posterior distribution $p(\theta|y_1, \dots, y_n)$ is almost surely concentrated at $\pm \sqrt{a^2-4}$ with equal probabilities $1/2$.

		(c) Under the same condition in (b), for any $\eta>0$, almost surely the following limits hold,
		$$ \limsup_{n\to \infty} \Pr\left( \left|\mu - \sqrt{a^2-4}  \right| < \eta  \mid y_1, \dots, y_n  \right) = \limsup_{n\to \infty} \Pr\left( \left|\mu + \sqrt{a^2-4}  \right| < \eta  \mid y_1, \dots, y_n  \right) =1 $$
	\end{theorem}
	When  $a>2$, if   $0.5 < p_0 \leq  \xi(a)$,  two modes ($\gamma^+, \gamma^-$) exist, but the exact inference will asymptotically concentrate at the right mode $\gamma = \gamma^+ $. Even when  $p_0=0.5$ so that the two centers $\pm a$ are equally important in the data generating process, the exact inference would still pick one mode, with the left and right mode having equal chances of being selected.
	
	\begin{corollary}\label{thm_elpd}
		In all cases in Theorem \ref{thm_post_converge}, the expected log predictive density (elpd) from the  exact Bayesian posterior $p(\mu|y_1, \dots, y_n)$ is 
		\begin{align*}
		\mathrm{elpd}_{\mathrm{Bayes}} =&\int_{\R} p_{\mathrm{true}}(\tilde y| p_0)  \log   \int_{\R} p(\tilde y|\mu) p(\mu|y_1, \dots, y_n)   d\mu  d \tilde y \\
		\xlongrightarrow{ n\to \infty } &  -\left(  p_0 \log  \left(  \pi  (  4+  (\gamma-a)^2 )  \right) + (1-  p_0 )\log  \left(  \pi  (  4+  (\gamma+a)^2 )  \right) \right) \\
		\overset {a \mathrm{~ is~large} }{\approx}    & -  (1- p_0 )\log  \left(1  +  a^2   \right)   - \log  4\pi.  
		\end{align*}
	\end{corollary}

	When  $a>2$, and  $0.5 \leq p_0 \leq  \xi(a)$,  the two modes ($\gamma^+, \gamma^-$) are detectable from multi-chain MCMC. In this case, stacking behaves better than exact Bayesian inference. Indeed, the next corollary shows that stacking approximates the data generating process in KL divergence.
	\begin{corollary}\label{thm_stacking}
		(a) When $n$ is large, for any $a >2$ and $0.5<p_0< \xi (a)$, both modes  $\gamma^-$ $\gamma^+$ receive asymptotically nonzero weights, and the elpd of the stacking average,
		$$
		\mathrm{elpd}_{\mathrm{stacking}} =\int_{\R}  p_{\mathrm{true}}(\tilde y| p_0)  \log   \int_{\R} p(\tilde y|\mu) p_{\mathrm{stacking}}(\mu|y_1, \dots, y_n)   d\mu  d \tilde y,
		$$
		is strictly larger than 	$\mathrm{elpd}_{\mathrm{Bayes}}$. 	
		
		(b)  When   $a$ is large, stacking weights for $(\gamma^-, \gamma^+)$ are asymptotically close to $1-p_0$ and $p_0$. Consequently, the stacked posterior predictive distribution  approximates the data generating process,
		$$\KL \left( p_{\mathrm{true}}(\cdot),   ~    \int_{\R} p(\cdot|\mu) p_{\mathrm{stacking}}(\mu|y_1, \dots, y_n)   d\mu  \right)   \gtrapprox 0, ~~  \mathrm{when~} n\to \infty, a \mathrm{~is~fixed~and~large}.$$
	\end{corollary}
	
	When $n$ is large, for $a >2, p_0=0.5$, the stacking weights for two modes $\pm \sqrt{a^2-4}$ are asymptotically equally 0.5.  We analytically evaluate the elpd under the true data generating process, the  asymptotic ($n \to \infty$) elpd  of full Bayes, and multi-chain-stacking in the right panel of Figure \ref{fig_xi_function}.  Stacking is predictively superior to the full Bayes. The elpd difference between the data generating process and stacking vanishes for a large $a$, implying the KL divergence between them approaches 0.
	
This Cauchy example at $p_0=0.5$ might remind readers of the one constructed by \citet{diaconis1986inconsistent}. They used a Dirichlet prior with the parameter measure Cauchy$(\mu, 1)$ to fit observations essentially coming from $y \sim 0.5 \delta{(a)} +0.5\delta (-a)$ with $a>1$.  The resulting Bayesian posterior of $\mu$ is concentrated at $\pm \sqrt{a^2-1}$.  However, instead emphasizing the inconsistency of this Bayesian procedure, we use our example to praise stacking: it approximates the \emph{true} data generating process given a \emph{misspecified} model, \emph{inconsistent} Bayesian inference, and \emph{non-mixing} samplers. The posterior multimodality is thereby a blessing rather than a curse under model misspecification.
	
	\section{Examples}\label{sec:example} 
We demonstrate the benefit of stacking by a series of multimodal posterior sampling tasks representing a range of challenging Bayesian computations.
	\subsection{Latent Dirichlet Allocation}\label{sec_lda}
\begin{wrapfigure}[9]{R}{0.6\textwidth}
        \small
        \vspace{-2em}
		\begin{tabular}{cl}
			{\small chain weight}& {\small top words in the topic}    \\ \hline
			0.20                & mr, man, wickham, good, give, young, lydia  \\
			0.18                & mr, man, young, bingley, collins, darcy      \\
			0.13                & mr, lady, catherine, dear, great, young      \\
			0.12                & wickham, elizabeth, mr, darcy, replied, hope \\
			0.09                & elizabeth, darcy, mr, sister, wickham, make 
		\end{tabular} \vspace{-0.5em}
		\caption{\em  Weights of the top 5 chains in the LDA model with $L=5$,
			and  top words in the topic that the first paragraph belongs to computed from these 5 chains.
		}\label{fig_lda_word}
			\end{wrapfigure}

	Latent Dirichlet allocation \citep[LDA,][]{blei2003latent} is a mixed-membership  clustering model widely used in natural language processing, computer vision, and population genetics. In the model, the $j$-th document ($1\leq j \leq J$) is drawn from the $l$-th topic ($1\leq l \leq L$) with probability  $\theta_{jl}$, where the topic is defined by a vector of probability distribution $\phi_l$ over the vocabulary, such that each word in the document from topic $l$ is independently drawn from a multinomial distribution with probability $\phi_l$.
	
Despite its popularity for data exploration, LDA suffers from computational instability, as the inference may not replicate itself from either multiple runs \citep{mantyla2018measuring} or data shuffle \citep{agrawal2018wrong}. This confuses users as a different result is produced from each new run,  and reduces the predictive power of text mining classifiers. Some literature recommends examining and selecting one best fit from multiple unstable inference results subjectively or through cross validation, or manually tuning hyperparameters to get rid of posterior multimodality, which however changed the original model and could further undermine classification efficiency  \citep{tian2009using,carreno2013analysis}. 
	
	We apply an LDA topic model to texts in the novel {\em Pride and Prejudice}. After removing frequent and rare words, the book contains 2025 paragraphs and 32877 words, with a total unique vocabulary size of 1495. We randomly split the words in the data into 70\% training and 30\% test. The dimension of the parameters  $\theta$ and $\phi$ grows as a function of the number of topics $L$ by $2025\times L$ and $L\times 1495$ respectively.  We place independent Dirichlet$(0.1)$ priors on $\theta$ and $\phi$. We vary $L$ from 3 to 15, and for each fixed model we sample with Stan using 30 parallel chains initialized at random starting points with 2000 or 4000 iterations per chain. 
		\begin{figure}
		\centering
		\includegraphics[width=1\linewidth]{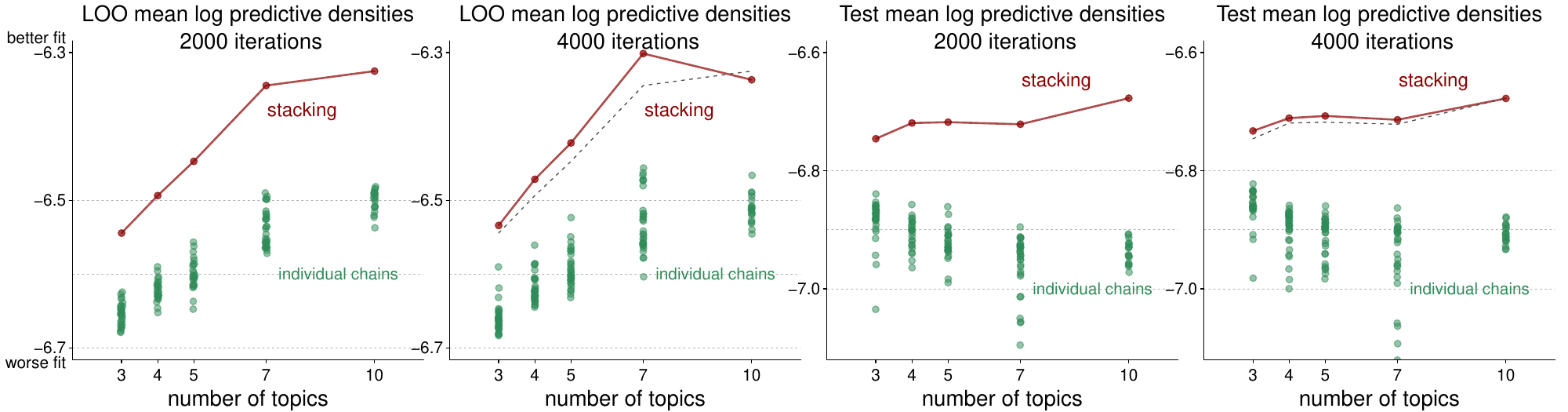}  \caption{\em The mean log predictive densities from 30 randomly initialized chains, and the stacked average of them, evaluated using both leave-one-word-out and independent test data. The number of topics $L$ in the LDA model  varies from 3 to 10, and each chain contains 2000 or 4000 iterations. Individual chains do not mix, and the best of them is invariably worse than stacking.}\label{fig_lda_lpd}
	\end{figure}
	
	Due to the well separated multimodal posterior $p(\phi, \theta|y)$, individual chains do not mix if they are run for more iterations. As represented by green dots in Figure \ref{fig_lda_lpd}, different  chains yield different log predictive densities on test data,
	suggesting the multimodality is more than label-switching. Figure \ref{fig_lda_word} lists, for five runs, the top words in the topic to which the first paragraph belongs.
	
	Following our stacking approach, the 30-chain-stacked average (red line in Figure \ref{fig_lda_lpd}) improves the model fit compared with even the best of individual chains by orders of magnitude, measured in  test data mean log predictive densities.  Indeed, the improvement of  stacking in mean lpd  ($\approx 0.2$) is standardized by sample size and  equivalent to roughly an $\exp(10^5)$ outperforming margin in the scale of Bayes factors.  There is a mismatch between the trend of loo and test lpd, indicating the inconsistency of single-chain loo-selection. This may come from (a) the non-iid nature of textual data, and (b) the parameter size is nearly the same as sample size, such that loo has not reached its consistency territory. But even so, stacking still performs well in test data and can be combined with other predictive metrics such as leave-one-document-out.

		\begin{figure}
	\centering
	\begin{subfigure}[t]{0.35\textwidth}
		\includegraphics[width=\linewidth] {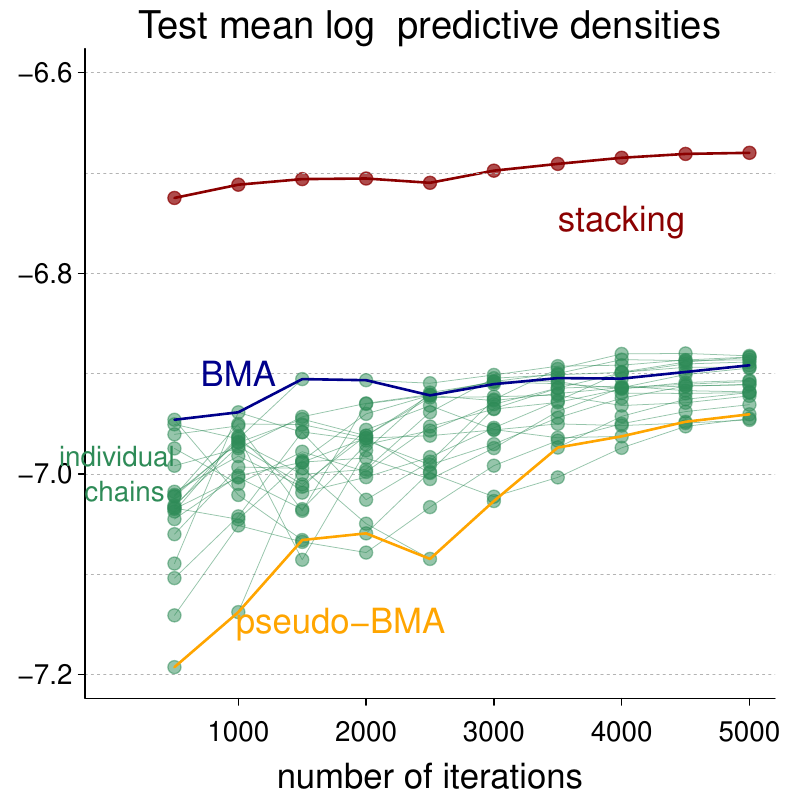} 
	\end{subfigure}
\quad \quad
	\begin{subfigure}[t]{0.35\textwidth}
		\includegraphics[width=\linewidth] {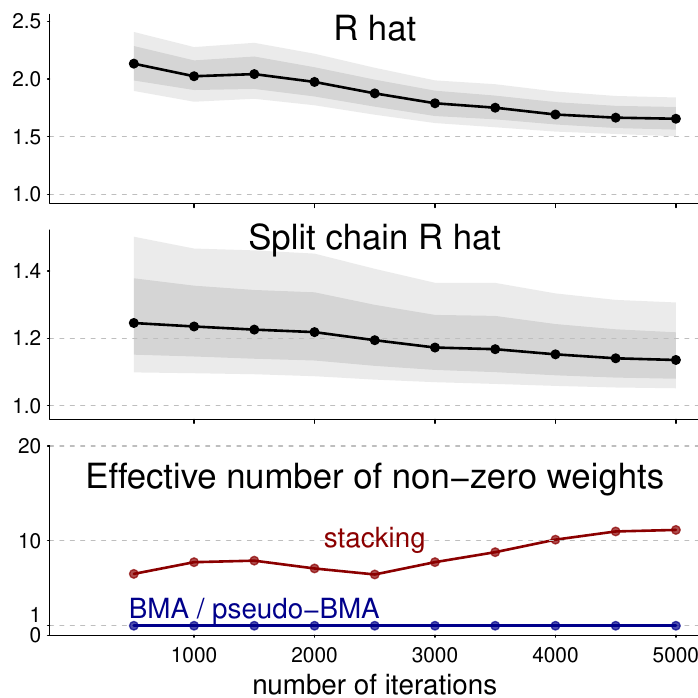} 
			\end{subfigure}
\linespread{1} 	\caption{ \em Stacking benefits from early-stopped MCMC.  We run LDA   with $L=10$ topics on 30 chains. As the number of iterations increase from 500 to 5000, the test  lpd of individual chains increases, while the stacked average has a flatter slope,  indicating we can stop early without losing much predictive power.  Monitoring $\widehat R$ and split-chain $\widehat R$ of all pointwise likelihoods, we find that $\widehat R$ is much bigger than split-$\widehat R$.  The bottom right shows the effective number of nonzero weights. BMA and pseudo-BMA put nearly all weight on one chain.}\label{fig_early_stop}	
	\end{figure}
	The left panel of Figure \ref{fig_early_stop} shows the test data predictive performance using a varying number of iterations from 500 to 5000 (with a fixed number of topics $L=10$).   As the number of iterations increases, test lpd from inferences using individual chains elevates, while the stacked average has a flatter slope, indicating that we can stop earlier and stack chains without losing much predictive power, even though these chains are not completely mixed. The upper and middle right panel show median, 30\% and 50\% central interval of $\widehat R$ and split-chain $\widehat R$ for all pointwise likelihoods. $\widehat R$ is much bigger than split chain $\widehat R$, suggesting that the non-mixing is mostly due to a lack of between-mode transitions. Given that in this problem sampling takes up to 12 hours CPU time per chain per 1000 iterations, such \emph{early stopping of iterations} provides a remarkable opportunity to reduce computation costs.  This is also manifested in Figure \ref{fig_lda_lpd}: for all $L\in [3,10]$,  individual chains perform better when per-chain iterations increase from 2000 to 4000, whereas the stacked average remains nearly unchanged (compare the red and dashed grey lines in the second and fourth panel).
	
	The bottom right panel of \ref{fig_early_stop} shows the effective number of nonzero weights.  In agreement with our theoretical discussion, BMA and pseudo-BMA put nearly all mass onto one chain, and  in fact they often do not even select the optimal chain for the test data (left column). Accordingly, it is no surprise that stacking outperforms BMA and pseudo-BMA.
	
	In addition to the benefit of early stopping of iterations, stacking provides an extra bonus of  \emph{early stopping of topics}.  Usually, the number of topics $L$ involves manual tuning. {\em Stacking effectively expands the model space}. Therefore, we observe in the right two panels of Figure \ref{fig_lda_lpd} that the stacked average is less sensitive to $L$ in test data lpd.  Stacking compensates for the lack of mixture components in the model through additional mixtures of posteriors during chain aggregation.

	\subsection{Gaussian Process Regression}
	
	Consider a regression problem with scalar observations $y_i = f(x_i)+\epsilon _i, i = 1,...,n,$ at input locations $X= \{x_i\}_{i=1}^n$, and $\epsilon _i$ are independent noises.  We place a Gaussian process prior on latent functions $f$ with zero mean and squared exponential covariance.
	In the next two experiments, we apply stacking to remedy bimodality in hyperparameter \emph{optimization}, and slow mixing in \emph{sampling}, respectively.
	
	\paragraph{Combining  modes in hyperparameter optimization.}\label{sec_gp_map}
	In Gaussian process regression, posterior bimodality can occur even with a normal likelihood: 
	\begin{equation}\label{eq_gp}
	y_i = f(x_i)+\epsilon _i, ~  \epsilon _i \sim  \mbox{normal}(0, \sigma),  ~ f(x) \sim  \mathcal{GP} \left( 0, \alpha^2 \exp\left( -\frac{(x-x')^2}{\rho^2} \right) \right).
	\end{equation}
	We use data from \citet{neal1998regression}. The univariate input $x$ is distributed $\mbox{normal}(0,1)$, and the corresponding outcome $y$ is also Gaussian  with standard deviation 0.1.  With probability 0.05, the point is considered an outlier and the standard deviation is inflated to 1. In all cases, the true mean of $y|x$ is
	\begin{equation} \label{eq_gp_f}
	f_{\mathrm{true}}(x) = 0.3 + 0.4 x + 0.5 \sin(2.7x ) + 1.1 / (1+x^2).
	\end{equation}
	
		\begin{figure}
		\centering
		\includegraphics[width= 0.97\linewidth]{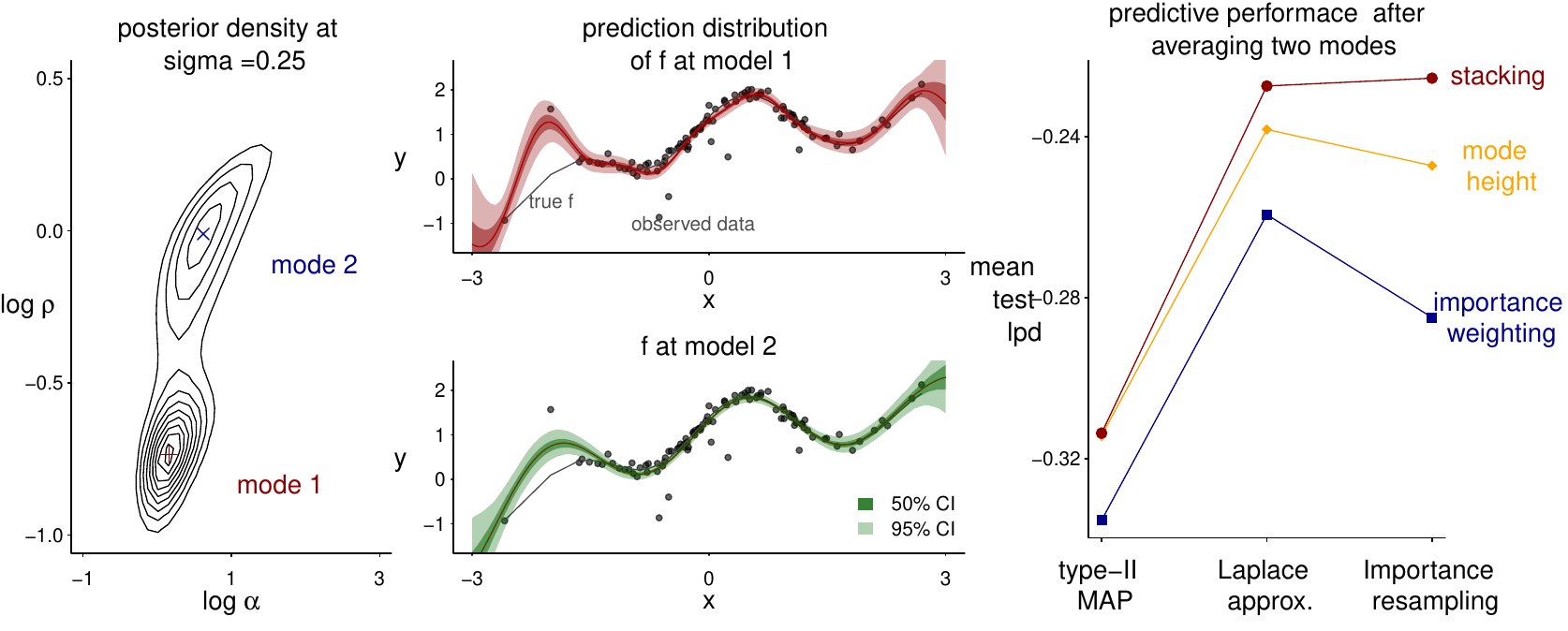}  \caption{\em The posterior distribution of hyperparameters $p(\rho, \alpha, \sigma|y)$ has at least two local modes. The left panel shows contours of the marginal posterior  of  $\rho$  and  $\alpha$ at fixed $\sigma=0.25$.  The middle panel shows draws from the  posterior predictive distribution $f|y$ at the two hyperparameter modes.  We can either pick these two modes as type-II MAP or  locally approximate the posterior of hyperparameters at the modes by Laplace approximation or uniform-grid importance resampling.  Then the resulting modes or local approximation can be combined according to stacking, mode height, or importance weighting. The right panel shows that stacking performs the best on test data log predictive densities for all schemes.}\label{fig_gp_MAP}
	\end{figure}
	
	Model \eqref{eq_gp} requires inference on $f(x_i)$ and all hyperparameters $\theta= (\alpha, \rho, \sigma)$.  
	We integrate out all $f(x_i)$ and obtain the marginal posterior distribution
	\begin{equation} \label{eq_gp_marginal}
	\log p(\theta | y) = - \frac{1}{2}  y^T \left( K(X, X ) + \sigma^2 I \right)  ^{-1} y  - \frac{1}{2} \log | K(X, X ) + \sigma^2I   | + \log p(\theta) + \mathrm{constant},
	\end{equation}
	where $p(\theta)$ is the prior for which we choose an elementwise $\mathrm{Cauchy}^+(0,3)$.
	
	In Neal's dataset with training sample size $n=100$, at least two local maxima of \eqref{eq_gp_marginal} can be found. We visualize the marginal distribution of $p(\rho, \sigma | y)$ at $\sigma =0.25$ on the leftmost of Figure \ref{fig_gp_MAP}.

	Now we consider three standard mode-based approximate inferences of  $\theta | y$:  
	
	\emph{a. Type-II MAP}. The value $\hat \theta$ that maximizes the marginal distribution \eqref{eq_gp_marginal}  is called the type-II MAP estimate.  Using this point estimate of hyperparameters $\theta= \hat \theta $, we further draw $f| \hat \theta, y$. 
	
	\emph{b. Laplace approximation}.  We compute $\Sigma$: the inverse of the negative Hessian matrix of \eqref{eq_gp_marginal} at the local mode $\hat \theta$, draw $z$ from multi-variate-normal$(0, I_{3})$, and use $\theta(z)= \hat \theta + \mathrm{V} \Lambda ^{1/2} z$ as the approximate posterior samples around the mode $\hat \theta$,  where the matrices $\mathrm{V},  \Lambda$ are from the eigendecomposition $\Sigma= \mathrm{V} \Lambda ^{1/2} \mathrm{V}^T$.
	
	\emph{c. Importance resampling}.  Instead of standard Gaussians in the Laplace approximation, we now draw $z$ from uniform$(-4,4)$, and then resample $z$ without replacement with probability proportional to  $p\left( \theta (z) | y\right) $ and use the kept samples of $\theta(z)$ as an approximation of $p(\theta|y)$.
	
	In the existence of two local modes $\hat \theta_1, \hat \theta_2$, we either obtain two MAPs, or two  nearly nonoverlapped draws, $(\theta_{1 s})_{s=1}^S,  (\theta_{2s})_{s=1}^S $. We then evaluate the  predictive distribution of $f$, $p_k(f|y, \theta)= \int\! p(f|y, \theta) q(\theta| \hat \theta_k) d \theta,~  k=1,2,$
	where $q(\theta| \hat \theta_k)$ is a delta function at the mode $\hat \theta_k$, or the draws from the Laplace approximation and importance resampling that is expanded at $\hat \theta_k$. We visualize the  predictive distribution of $f$ using two local MAP estimates in the middle panel of Figure \ref{fig_gp_MAP}. The one with the smaller length scale is more wiggling and passes the training data more closely. 
	
	For each of these three mode-based inferences, we consider three  strategies to combine two modes:
	
	\emph{a. Mode height}. We reweigh the predictive distribution of $f$  according to the height of the marginal posterior density at the mode: $w_k \propto p(\hat \theta_k| y), k=1,2.$  
	
	\emph{b. Importance weighting}.   For approximate posterior draws $(\theta_{1 s})_{s=1}^S,  (\theta_{2s})_{s=1}^S $, we reweigh them  proportional to the mean marginal posterior density $w_k \propto 1/S \sum_{s=1}^S  p(\theta_{ks} | y )$. We choose the importance weights of two MAPs using the ones from importance resampling as it approximates the total posterior mass in the surrounding region near the mode. 
	
	\emph{c. Stacking}.  Our fast approximate loo does not apply to MAP estimation directly. Therefore, we split the data into training $y_\mathrm{train}$ and validation data $y_\mathrm{val}$.  We first obtain either MAPs or approximate hyperparameter draws using training data and optimize their predictions  on validation data. Stacking maximizes  $\sum_{i=1}^{n_\mathrm{val}} { \log \left( \sum_{k=1}^K w_k p(y_{\mathrm{val}, i}| y_\mathrm{train} , \hat \theta_k)\right)  }$ for MAPs or 
	$\sum_{i=1}^{n_\mathrm{val}} { \log\left(  \frac{1}{S}  \sum_{k=1}^K  w_k  \sum_{s=1}^S  p(y_{\mathrm{val}, i}| y_\mathrm{train} ,   \theta_{ks}) \right) }$ for Laplace and importance resampling draws.
	
	In the right panel of Figure \ref{fig_gp_MAP}, we evaluate these three weighting strategies by computing the mean expected log predictive density of the combined posterior distribution on hold-out test data (${n_\mathrm{test}=300}$).  No matter whether we are combining two point-estimates or two distinct Laplace/importance resampling draws near the two modes, the stacking weights provide better predictive performance on test data.
	
	\paragraph{Combining  non-mixed chains  from  Gaussian process regression with a Student-$t$ likelihood.}
	
	\citet{neal1998regression} originally constructed this example in which noise $\epsilon_i$ in \eqref{eq_gp} is modeled by   a $t$ distribution with mean 0, scale $\sigma$ and  degrees of freedom  $\nu$:
	$$p(y_i| f_i, \sigma, \nu)= \frac{\Gamma((\nu+1/)2) }{\Gamma(\nu/2) \sqrt{\nu\pi}\sigma }   \left(  1+ \frac{(y_i-f_i)^2}{\nu \sigma^2} \right)^{-{(\nu+1)/2}}, ~ f \sim  \mathcal{GP} \left( 0, \alpha^2 \exp\left( -\frac{(x-x')^2}{\rho^2} \right) \right).$$
	The   Student-$t$ model is robust to outlying observations but is computationally challenging, because of (a) lack of closed-form expression for $p(f|y)$, and (b) heavy-tailed posterior densities. Approximate methods exist, such as factorizing variational approximation \citep{tipping2005variational}, Laplace approximation \citep{vanhatalo2009gaussian}, and expectation propagation  \citep{jylanki2011robust}, but posterior sampling remains difficult.    
	
	We generate training  data $x_{1:n}$ from uniform$(-3, 3)$, and the outcome $y_i$  has the same mean in  \eqref{eq_gp_f}. $y_i$ either has standard deviation $\sigma_1=0.1$, or  inflated to $\sigma_2>0.1$ with probability  proportional to $\exp\left(  -   \left(C  \frac{  i - 0.4n  } {n }\right) ^2\right),$ where $C>0$ is a concentration factor that decides how the outliers are concentrated with each other in $x$-space.  In the experiment, we vary $\sigma_2$ from $0.1$ to $1$ and $C$ from 1 to 8.    $n_\mathrm{test}=300$ hold-out test data points $(\tilde X_i, \tilde y_i)_{i=1}^{n_\mathrm{test}}$ are generated from the same mechanism.
		
	\begin{figure}
		\centering
		\includegraphics[width=1\linewidth]{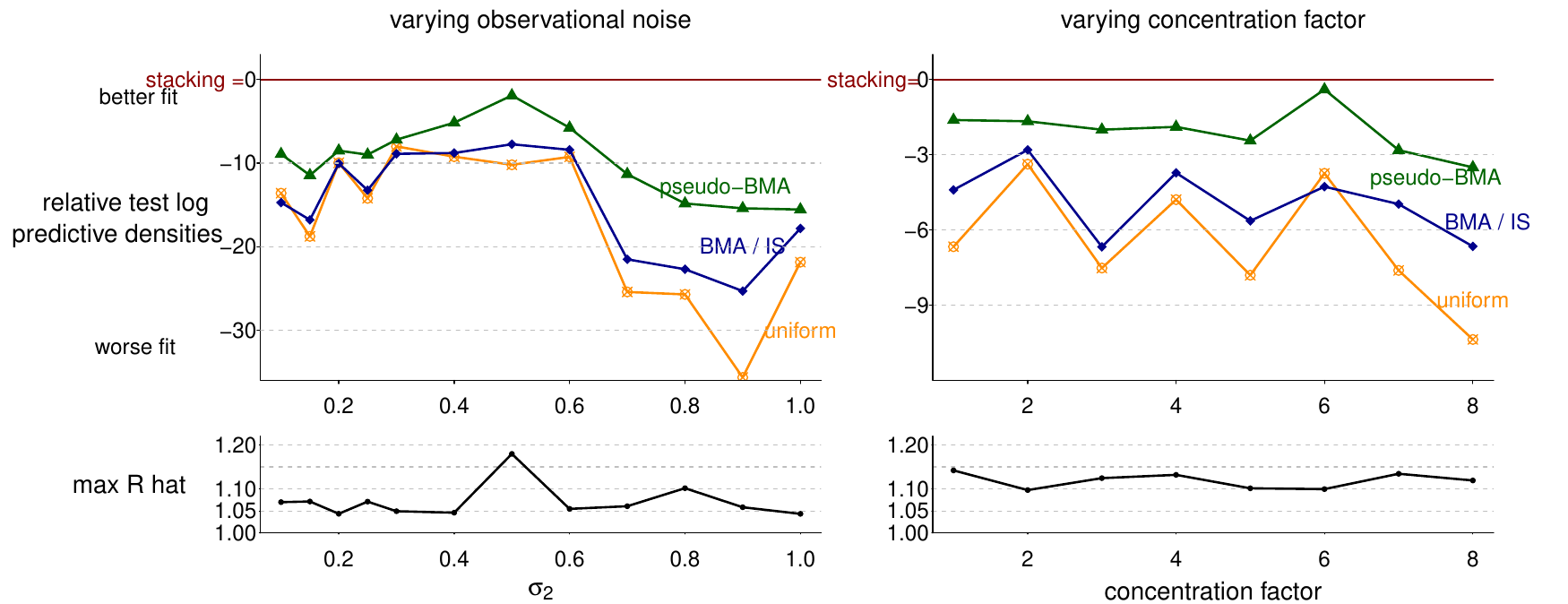} \vspace{-1em}
		   	\caption{\em    Left: We fix the concentration factor $C=5$ and vary the outlier standard deviation $\sigma_2$ from $0.1$ to $1$ in the data generating mechanism. Right: We fix $\sigma=0.3$ and vary  the concentration factor $C$ from 1 to 8. In each setting,  we sample from the posterior distribution using 8 chains with 8000 iterations each, and combine chains using four weighting methods. We  report the test log predictive densities (using $n_\mathrm{test}=300$ independent test data) of three other methods subtracting stacking, which are always negative. The lower row reports the maximum  $\widehat{R}$ among all parameters.}\label{fig_gp}
	\end{figure}

	We fix the degrees of freedom $\nu= 2$ and sample from the full posterior distribution $p(f_1, \dots, f_n, \sigma, \alpha, \rho)$ from $K=8$ parallel chains and 8000 iterations per chain in Stan. We  draw initialization from uniform$(-10,10)$ for unconstrained parameters and set the maximum tree depth to 5 in the  No-U-Turn sampler \citet[NUTS;][]{hoffman2014no}.   In the lower row of Figure \ref{fig_gp}, we report the maximum $\widehat{R}$ of all sampling parameters among 8 chains:  clearly not mixing.
	
	We compare four chain-combination strategies:  BMA, pseudo-BMA, uniform averaging, and stacking.  After each iteration of $(\sigma, \rho, \alpha, f)$, we draw  posterior predictive sample of $\tilde f= f(\tilde X)$, 
	and compute the mean test data log predictive densities. 
	Since test performance changes in orders of magnitude under different data-generating settings, in Figure \ref{fig_gp} we use stacking as a baseline and compare the test log predictive densities of other methods by subtracting stacking ones. In all cases, stacking outperforms the other three approaches. 
    \begin{figure}[!h]
	\centering
	\includegraphics[width=1\linewidth]{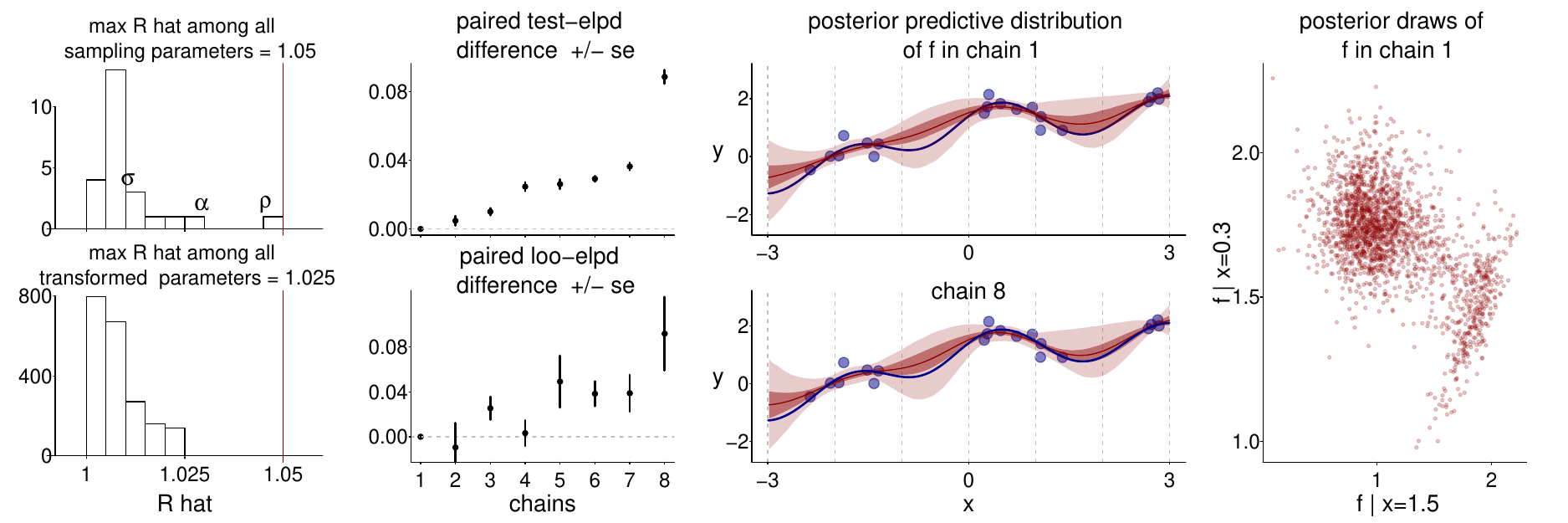} \vspace{-1em}
	\caption{\em    In this experiment with $n=20, \sigma_2=0.6, C=5$, even when $\widehat{R}$ for all  parameters are smaller than 1.05,  the 8  chains exhibit different predictive capabilities. The second column shows the estimated log predictive densities subtracting chain $1$ and the standard error in test data or loo.  Chains have been reordered by test scores. The third column shows the prediction of $f$ in chains 1 and 8. The rightmost column is the joint posterior predictive draw $f$ at $x=1.5$ and 0.3 in chain 1.}\label{fig_gp_20}
\end{figure}

	There are three contributors to the poor mixing in this example. First, chainwise predictions may diverge even when  parameters are nearly mixed. Figure	\ref{fig_gp_20} display sampling results for a dataset with $n=20, \sigma_2=0.6, C=5$. In the leftmost column, all $(\sigma, \rho, \alpha, f)$ and transformed parameters have  $\widehat{R}<1.05$.  But the log predictive densities are different across chains, shown in the second column (chains have   been re-ordered by test lpd). Stacking is a more powerful diagnostics tool in this case.  
	
	Second, the posterior distribution $f|y$ can be multimodal.  The rightmost column of Figure \ref{fig_gp_20} displays the joint bimodal posterior distribution of $f$ conditioning on $x=0.3$ and $1.5$.  In this example, this is not a sampling concern owing to the small between-mode energy barrier, and HMC/NUTS sampler in Stan is able to move between these two modes rapidly.
	\begin{figure}
	\centering
	\includegraphics[width=1\linewidth]{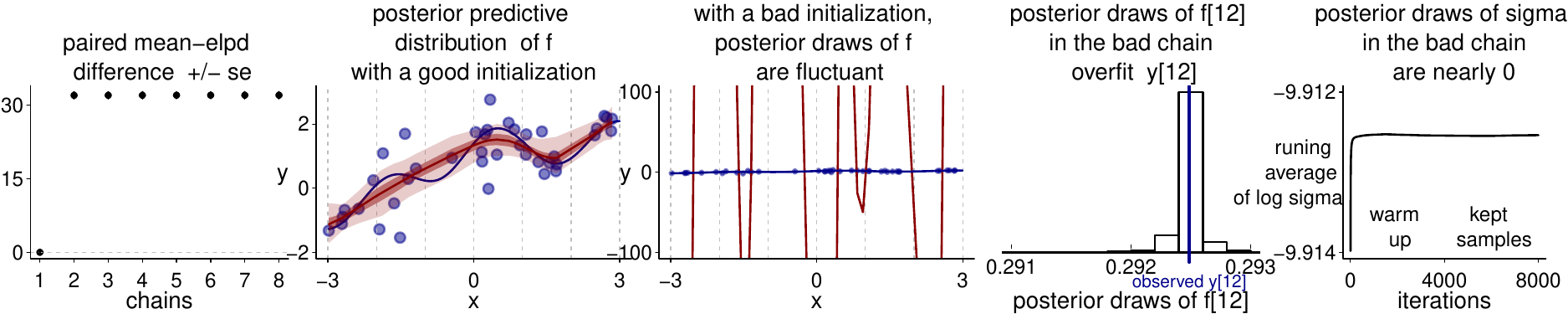}
	\caption{\em  In  an experiment with $n=40, \sigma_2=1, C=5$, chain 1 is trapped in a bad local mode, overfitting the data, and with $\sigma$ is trapped near zero for 8000 iterations.  This chain has a low elpd on both test data and loo, hence is abandoned in stacking.}\label{fig_gp_40}
\end{figure}

	Third, some chains may be trapped in bad local modes. In Figure \ref{fig_gp_40}, we outline the sampling result from another dataset ($n=40, \sigma_2=1, C=5$). Chain 1 is trapped in a local mode with $\sigma \approx 0$ and is unable to escape the local trap after 8000 iterations.  The posterior prediction $f$ fluctuates and overfits the observations:  $f_{12}|y$ is nearly a delta function at $y_{12}$.  The strong overfitting of this chain leads to a low elpd on both test data and leave-one-out cross validation, hence it is abandoned by stacking. 
	
	
	\subsection{Hierarchical Models}
	\paragraph{When the bimodality occurs and when reparameterization helps.}\label{sec_hier_1}
	Consider observations from  $J$ exchangeable groups. For simplicity we assume a balanced one-way design, with data $y_{ij}, i=1, \dots, N$,  from groups $j$. We apply a hierarchical model with  parameters $\left( \theta, \sigma, \mu, \tau\right)$,
	\begin{equation}\label{eq_cp}
	\mathrm{centered}:  \quad  y_{ij}  | \theta, \sigma \sim \mbox{normal}(\theta_j, \sigma), ~ \theta_j |\mu, \tau \sim \mbox{normal}(\mu, \tau), ~  1\leq i\leq N, ~ 1\leq j \leq J.
	\end{equation}
	Sampling in the space of  $\left( \theta, \sigma, \mu, \tau\right)$ is called \emph{centered parameterization}.  When the likelihood is not strongly informative, the prior dependence between $\tau$ and $\theta$ in \eqref{eq_ncp} can produce a funnel-shaped posterior that is non-log-concave, and slow-to-mix near $\tau=0$.
	
	Alternatively, with \emph{non-centered parameterization}, we sample $\left( \xi, \sigma, \mu, \tau\right)$ through  a bijective mapping $ \theta_j =\mu+ \tau \xi_j$, and the model is equivalently reparameterized by 
	\begin{equation}\label{eq_ncp}
	\mathrm{non\!\!-\!\!centered}:  \quad y_{ij} \sim \mbox{normal}(\mu+ \tau \xi_j, \sigma), ~ \xi_j \sim \mbox{normal}(0, 1), ~  1\leq i\leq N, ~ 1\leq j \leq J.
	\end{equation}

	When the likelihood is not strongly informative, the non-centered parameterization is preferred \citep{betancourt2015hamiltonian, gorinova2019automatic}, 
	but when the likelihood is strongly informative, then the non-centered parameterized posterior has a funnel shape.
	The data informativeness can be crudely measured  by 
	the inverse of $F$-statistics (between group variance divided by within group variance). But beyond such heuristics and  limited classes of models where analytic results can be applied, there is no general guidance on which parameterization to  adopt.
	
	Parallel to the slow mixing rate due to the funnel-shaped posterior, the posterior in \eqref{eq_cp} can contain two modes, usually arising when the data indicate a  larger between‐group variance than does the prior.  \citet{liu2003posterior} characterized the bimodality of this model under conjugate priors in closed form. 
	
	To understand how the posterior bimodality  affects sampling efficiency, in the first simulation we generate data from $J=8$ groups and $N=10$ observations  per group.  The true $\tau$ and $\sigma$ vary from $0.1$ to $20$, with a varying amount of $t$-distributed noise added to $\theta$. 
	We place independent conjugate inverse-gamma$(0.1, 0.1)$ priors on $\tau^2$ and $\sigma^2$. For every  realization of data, we sample from the posterior distribution in both centered and non-centered parameterization using 4000 iterations, and analytically determine whether the centered parameterization has two posterior modes.

		\begin{figure}
		\centering
		\includegraphics[width=1\linewidth]{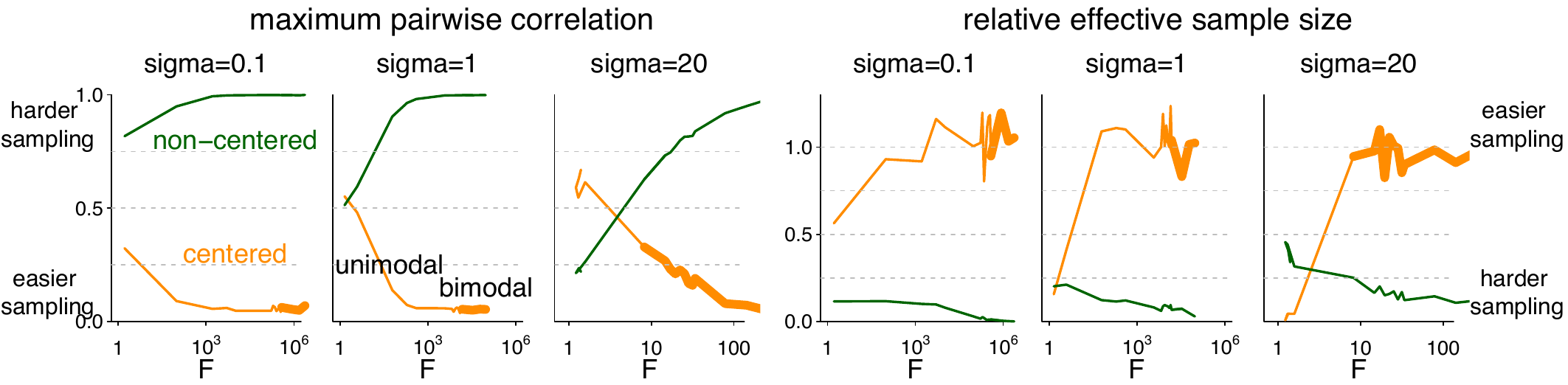} \vspace{-0.5em}
		\caption{\em  We fit the hierarchical model on data simulated from by various generating process.  When the between group variation is large or the within group variation is small, whose ratio is the sample F statistics, the centered parameterization is more efficient, amid less correlated posterior and large effective sample size.  Counterintuitively, this is also when the posterior bimodality occurs.}
		\label{fig:cp_vs_ncp}
	\end{figure}
	\begin{figure}[!ht]
		\centering
		\includegraphics[width=1\linewidth]{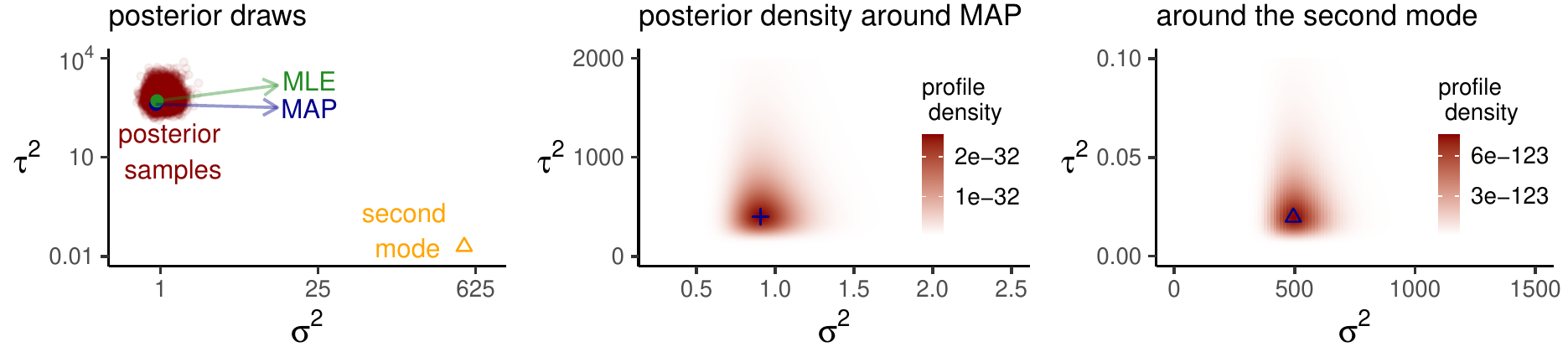} \vspace{-0.5em}
		\caption{\em A grid search finds two posterior modes  when data are generated by $\sigma=1$ and $\tau=25$. The second mode  in density and prediction ability is ignored by posterior sampling.}
		\label{fig:randrom_eff_bimodal}
	\end{figure}
	
	In Figure \ref{fig:cp_vs_ncp}, we assess the maximum absolute parameterwise correlations (left three columns), and the relative effective sample size (ESS divided by total iterations, right three columns) in  posterior samples.  Conforming to our heuristics, when between-group variation is large and within-group variation is small,  the centered parameterization is more efficient, and vice versa.

	Surprisingly, in this example metastability and multimodality evolve in opposite directions.  In Figure \ref{fig:cp_vs_ncp} we visualize the occurrence of posterior bimodality in centered parameterization by thicker line width. When the between-group variation increases, the centered posterior eventually becomes bimodal, but sampling becomes more efficient.

	How is this possible?  Figure  \ref{fig:randrom_eff_bimodal} presents an example where the data are generated by  $\sigma=1$ and $\tau=25$.   Both the MAP and MLE are close to the true value. A second local mode explains all variation by a large $\sigma$ (opposite to Figure \ref{fig_gp_40}), but it is orders of magnitude lower than the first one in posterior densities, hence ignored by  sampling. That's why the centered parameterization runs smoothly in the existence of posterior bimodality. The bad mode also has a low loo elpd, so stacking assigns it zero weight when we combine the modes.

	\paragraph{A stacked parameterization and zero-avoiding priors.}\label{sec_hier_2}
	Section \ref{sec_hier_1} leaves a few open problems:  which parameterization to choose in practice, whether the sample has included all local modes, whether the ignored modes are predictively important, and if we should search for them in the first place. The bimodality analysis of \citet{liu2003posterior} applies to conjugate priors. But  multimodality readily exists in hierarchical models. When the group-level standard deviation $\tau$ has a flat prior, $\tau=0$ is \emph{always} a mode of the joint posterior distribution.  From the modeling perspective, this mode represents complete pooling.   
	
	Given that the centered parameterization behaves like an implicit truncation and  has sampling difficulty in the small $\tau$ region, we propose a stacking-based solution for reparameterizations. We run $K+1$ chains. The first chain is complete pooling: restricting $\tau=0$ and $\theta_j=\theta_1$. The next $K$ parallel chains  are centered parameterization with a zero-avoiding prior  \citep{chung2013nondegenerate} on $\tau$.   Finally, we use stacking to average these $K+1$ chains.  Intuitively, if  $\tau \approx 0$ is predicatively important but missed by the  implicitly left truncated centered parameterization, the first chain fills the hole; when  $\tau \approx 0$ is incompetent, the centered sampling is boosted by circumventing the computationally intensive  region $\tau \approx 0$.
	\begin{figure}
		\centering
		\includegraphics[width=1\linewidth]{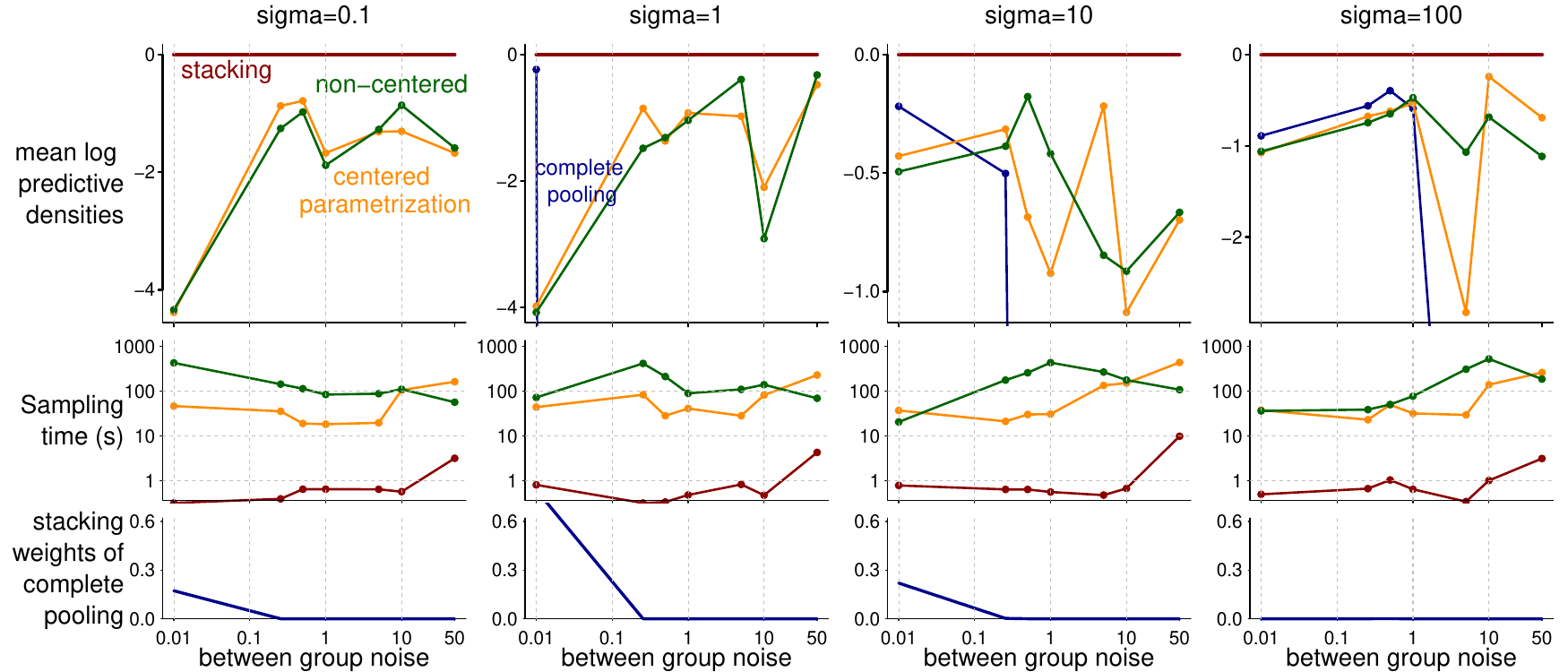}	 \caption{\em  We stack 8 parallel centered-parameterized chains and  1 complete pooling chain. The stacked average always has better test data performance than both centered and non-centered ones in all data configurations. The additional computation cost of stacking is minimal.  Even when the complete pooling chain receives zero weight, stacking still helps remedy slow mixing of remaining chains and achieves better elpd than uniform mixing.}	\label{fig:cp_stacking_full}
	\end{figure}
	
	To validate our proposal,  we simulate data with  dimensions $J=100$ (number of groups) and  $N=20$ (observations per group).   We vary the true within-group standard deviation $\sigma$ from $0.1$ to $100$ and add between-group noises $Bv_j$ to $\theta_j$, where $B$ is a   constant scalar varying from 0 to 50, and each $v_j$ is an independent Student-$t(1)$ noise.  We place  a zero-avoiding prior $\tau^2 \sim$ inv-gamma$(0.1, 0.1)$.
	We sample one chain (3000 iterations) from the complete pooling model, eight chains each from  centered and non-centered parameterization, stack the complete pooling and centered ones, and  evaluate the prediction ability of the posterior inference using mean log predictive densities on $N_\mathrm{test} =300$ independent test data in each group. In the upper row of Figure \ref{fig:cp_stacking_full}, we place the stacking average as the baseline and extract its elpd from other parameterizations.  The complete pooling model almost always has lpd so low that it does not even appear on the graph, and should never be used by itself.  Instead of picking between the centered or and non-centered  parameterization, the stacking estimate (red line) always has a larger log predictive density than the best of them.  Such advantage is achieved at a negligible computation cost compared with sampling time (middle row).  These  patterns are robust under different prior and data configurations, and we have omitted similar outcomes when we tune $J$ from 10 to 500 and for other zero-avoiding priors.
	
	Lastly, in this example, stacking remedies both the incapability to sample in small $\tau$ regions, and between-chain-non-mixing in the centered parameterization. The last row of Figure \ref{fig:cp_stacking_full} monitors  stacking weights for the complete pooling chain. Even when it receives zero weight, the stack-weighted draws from centered parameterization are better than the  uniform mixing of eight chains.
	
	\subsection{Stacking Multi-run Variational Inference in a Horseshoe Regression}\label{sec_example_vi}
	The regularized horseshoe prior \citep{piironen2016hyperprior, piironen2017sparsity} is an effective tool for Bayesian sparse regression. Denoting $y_{1:n}$ as a binary outcome and $x_{n\times D}$ as predictors, the logistic regression with a regularized horseshoe prior is,
	\begin{align*}
	\Pr(y_i=1)= \mathrm{logit}^{-1} ( \beta_{0}+\sum_{d=1}^{D} \beta_{d} x_{id} ), ~i=1,\dots, n, \quad  \beta_d | \tau, \lambda, c\sim  \mbox{normal}\left( 0,     \frac{\tau c\lambda_d}{(c^2 + \tau^2 \lambda_d^2)^{1/2}}\right), \\
	c^2 \sim  \mathrm{Inv\!\!-\!\!Gamma} (\alpha,\beta),   \quad \tau \sim \mathrm{Cauchy}^+(0,1), \quad  \lambda_d \sim \mathrm{Cauchy}^+(0,1),  ~ d=1, \dots, D. 
	\end{align*}
	
	
	Sampling from the exact posterior $p(\beta,\tau, c, \lambda| y)$ is computationally intensive and not scalable to big data.  Unfortunately, mean-field variational inference \citep[VI, ][]{blei2017variational} which optimizes over the best mean-field Gaussian approximation to the joint posterior measured in KL divergence, behaves poorly on horseshoe regression. 
	In particular, VI cannot capture the posterior multimodality \citep[see examples in][]{yao2018yes}, which is a key aspect of the regularized horseshoe, a continuous counterpart of the spike-and-slab prior.  
	
	In general, the optimization problem in variational inference is not convex. Equipped with stochastic gradient descent, multiple runs of variational inference can return entirely different parameters. The common practice is to either select the best run based on the evidence lower bound (elbo) or test data performance.  In the presence of posterior multimodality, the best that a normal approximation can do is to pick one mode, which in particular undermines the advantage of altering between no pooling and complete pooling of horseshoe regressions. 
	
	In the next two experiments, we apply stacking to multiple runs of automatic variational inference \citep[ADVI, ][]{kucukelbir2015automatic}.  In the $k$-th run, $k=1, \dots,K,$ we obtain $S$ posterior approximation draws $\theta_{k1}, \dots, \theta_{kS}$. We treat these as posterior samples, obtain the leave-one-out predictive densities, and use stacking to derive the optimal combination weights of all $K$ runs.

	\paragraph{Synthetic data.} 
	\begin{figure}
		\centering
		\begin{minipage}[b]{.4\textwidth}
			\centering
			\includegraphics[height=5.3cm]{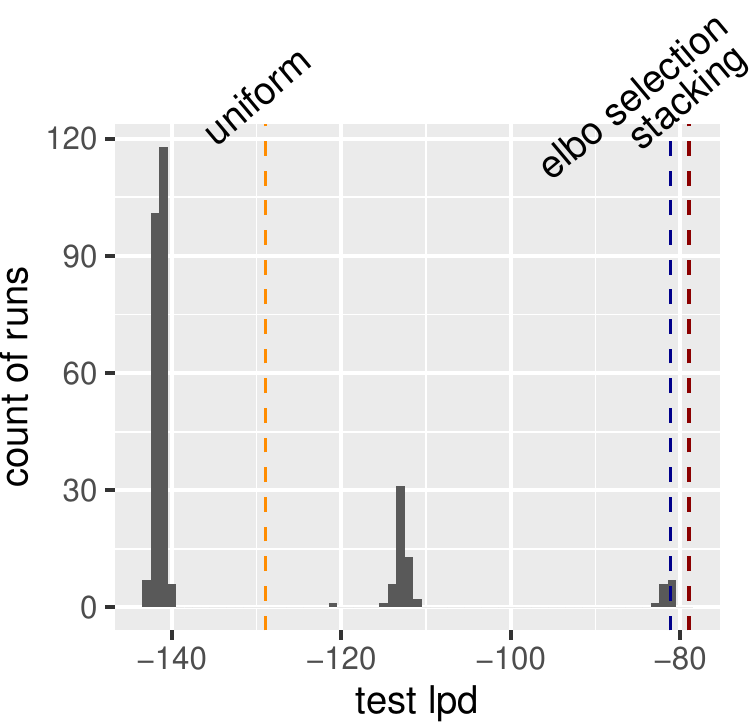} \caption{\em Test data elpd among 300 runs of variational inference using synthetic data. Stacking over 300 runs achieves better prediction than any single run, and also outperforms uniform mixing.}\label{fig:HS_test}
		\end{minipage}
		~
		\begin{minipage}[b]{.51\textwidth}
			\centering
			\includegraphics[height=5.3cm] {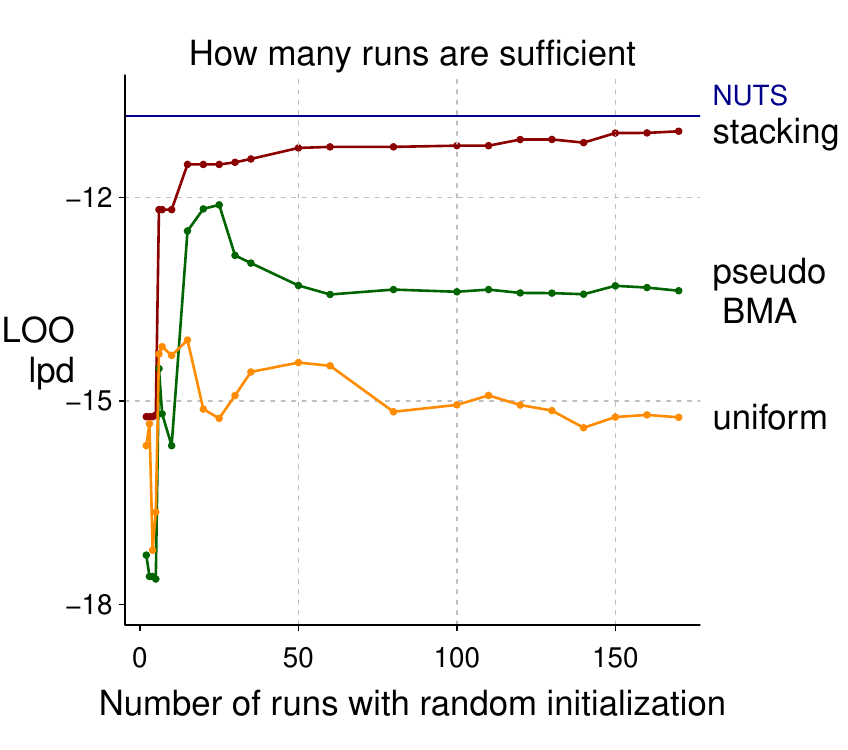}	 \caption{\em Monitoring convergence for the leukemia example.  Pseudo-BMA and uniform weighting have lower loo-lpd with more runs. Stacking is stable after 10 runs and gives a fit close to NUTS while requiring much less computation time.}\label{fig:HS_convergence}
		\end{minipage}
	\end{figure}

	We generate data from the model, $ \mathrm{Pr}(y_i=1)=   \mathrm{logit}^{-1}\! \left(   \sum_{d=1}^{400} \beta_d x_{id} \right)$, $i=1,\dots, n=40.$  The design matrix $X$ is normally distributed with shared featurewise components to increase  linear dependence. Of the 400 predictors, only the first three have nonzero coefficients $\beta_{1,2,3}=(3, 2, 1)$; this is the example discussed in \citet{van2014horseshoe} and \citet{piironen2017sparsity}.  We  assess the model prediction on hold-out test data with size $n_\mathrm{test}=200$. 
	
	
	Figure \ref{fig:HS_test} presents the test data log predictive densities among 300 ADVI runs with $10^5$ stochastic gradient descent iterations each run. Stacking achieves better prediction than any single run and uniform mixing.  Most of the runs have a low lpd, making the uniform reweighing undesired. The elbo selection selects the second-best run (in test data lpd).

	\paragraph{Leukemia classification.} 
	We consider regularized horseshoe logistic regression on the leukemia classification dataset. It contains 72 patients $y_i=0 ~ \mathrm{or}~1, 1\leq i\leq 72$, and a large set of predictors consisting of 7128 gene features $x_{id}, 1\leq d\leq 7128$.
	
	In this section, we view HMC/NUTS sampling in Stan as the gold standard, which is slow (several hours per 1000 iterations) but mixes well in this dataset \citep{piironen2017sparsity}.   We push the limit of variational inference by averaging 200 parallel ADVI  runs with $10^5$ stochastic gradient descent iterations, where each run takes less than one minute, but the approximation from any VI run is inaccurate. 
	
	Figure \ref{fig:HS_convergence} displays the leave-one-out log predictive density of the combined distribution as a function of the number of runs to average, as previously described in \eqref{eq_monitor}.  For stacking, there is a first jump at 5 runs, a second jump at roughly 10 runs, and then almost stable afterward.  For pseudo-BMA and uniform weighting, the loo elpd is worse with  more runs, because  VI is sensitive to initialization, and pseudo-BMA, BMA, and uniform weighting are sensitive to weak but duplicated runs \citep{yao2018using}. Stacking achieves a much better leave-one-out lpd than all individual chains and other weighting methods,  nearly comparable to HMC/NUTS. 
	There is one caveat:  because of the optimization procedure, 
	the loo lpd of stacking likely overestimates its expected lpd.  
	
	\begin{figure}
		\hspace{-1.3cm}
		\begin{subfigure}{0.5\textwidth}
			\centering
			\includegraphics[height= 6.2cm]{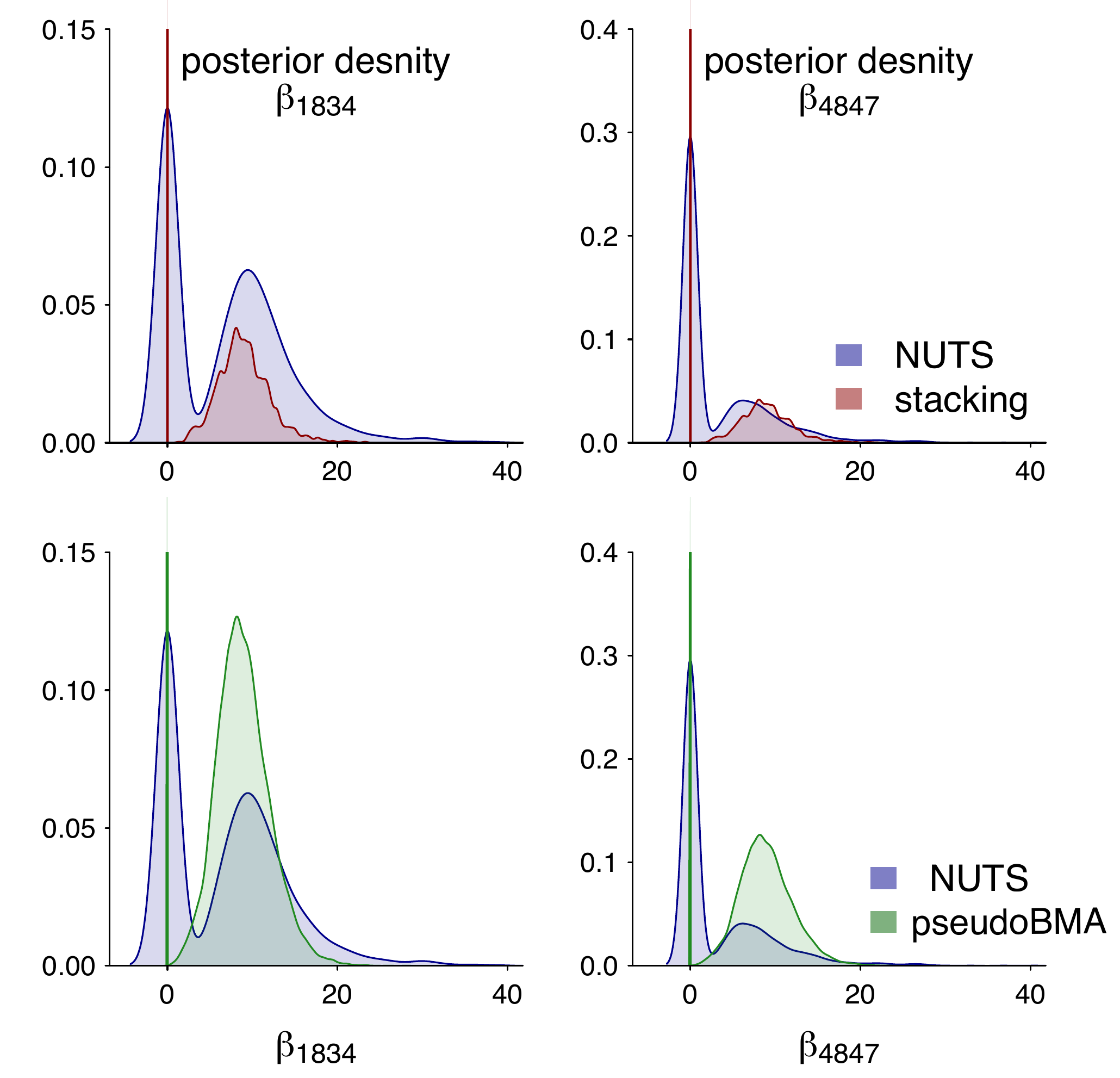}
		\end{subfigure}%
		~ 
		\hspace{-0.5cm}
		\begin{subfigure}{0.58\textwidth}
			\centering
			\includegraphics[height=6.1 cm]{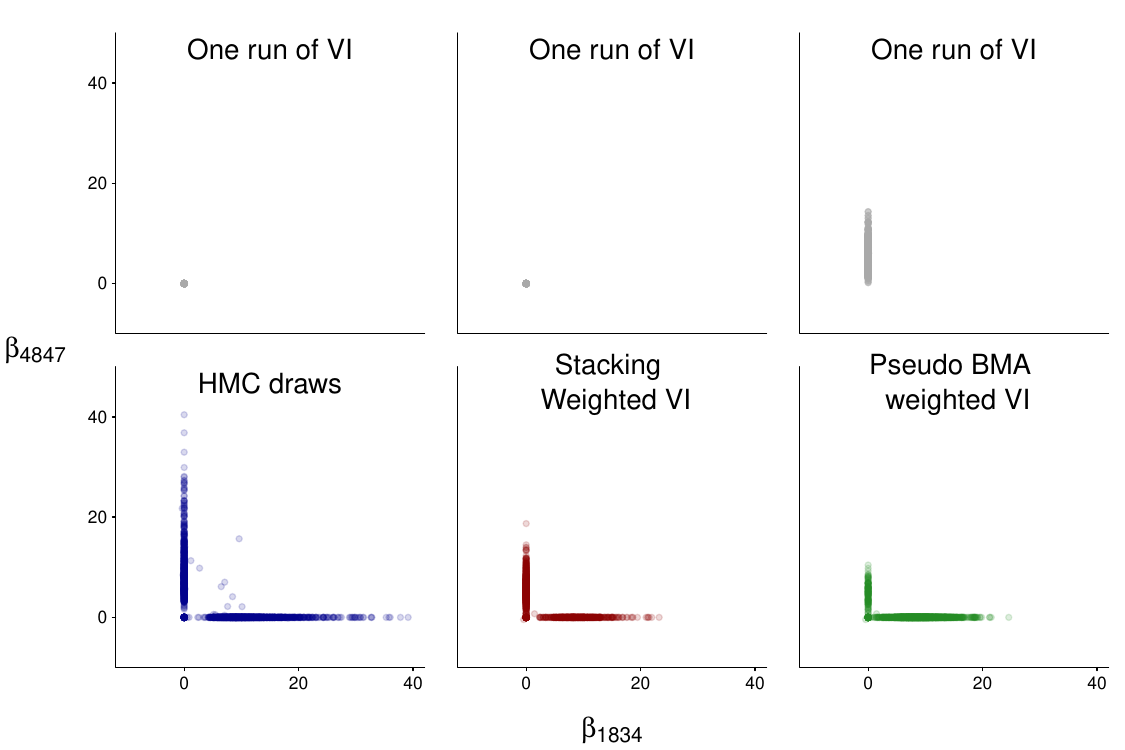}
		\end{subfigure}%
		 	\caption{\em The stacked VI posterior distribution matches HMC/NUTS draws  reasonably well both marginally (left panel) and jointly (right panel) for the leukemia example, although individual runs are inaccurate. The graph displays two parameters  $\beta_{1834}$ and  $\beta_{4847}$ that  have the largest absolute posterior means.}	\label{fig:horseshoe_density}
	\end{figure}
	
	
	To better evaluate how close the final inference is to the exact sampling,  we visualize the stacked posterior VI draws of $\beta_{1834}$ and  $\beta_{4847}$ (we pick these two variables which in our computation had the largest absolute posterior means as estimated with HMC/NUTS) in Figure \ref{fig:horseshoe_density}. Stacked VI approximates the posterior well both marginally (left two columns) and jointly (right three columns). It captures the main shape: a spike concentrated at 0 and a slab part---a true spike in the stacked distribution might be even more appealing for interpretation.  We also plot the joint distributions from three individual runs, all distant from the truth.  Stacking recombines these individual mean-field normal approximations,  the mixture of enough of which can approximate any continuous distribution.
	
	Finally as a caveat, the PSIS-loo approximation is applicable to VI under the assumption that each VI optimum $q_k$ locally matches the exact posterior $p$ (up to a normalization constant $c_k$):
	\begin{equation}\label{eq_VI_assum}
	\exists \Theta_k \subset \Theta, ~     q_k(\Theta_k) \approx 1,  \quad \mathrm{s.t.} ~ \forall \theta \in \Theta_k, ~   q_k(\theta) \approx c_k p(\theta|y),
	\end{equation}
	which can be assessed by diagnostics in  \citet{yao2018yes}.  In this example, it is implausible that \eqref{eq_VI_assum} would exactly hold, but PSIS-loo still yields useful results.  Alternatively, we can circumvent assumption \eqref{eq_VI_assum}, replace loo by a training-validation split, and perform stacking on the validation set, as shown in Section \ref{sec_gp_map}.
	
	\subsection{Bayesian Neural Networks}\label{sec_nn}
	
	The posterior distribution of neural network parameters is well known to be often multimodal. We demonstrate stacking for such an example using the MNIST dataset, a collection of images of handwritten digits that are to be classified into their true labels, 0--9.  We consider a two-layer neural network with tanh activation function:
	$$ \mathrm{Pr}(y_i = k) \propto  \exp(\sum_{j=1}^J  h_{ij}  \beta_{jk}  + \phi_k),   \quad  h_{ij}=  \mathrm{tanh}( \sum_{m=1}^{M}   x_{i m} \alpha_{m j}  ),   \quad  i=1, \dots, n, \ k=0, \dots, 9.$$
	where $n$ is the sample size,  $J$ is the number of hidden nodes, and $M=784$ is the input dimension.  Making scalable Bayesian inference remains  an open computation problem and beyond the scope of this paper. To simplify the problem while keeping the pathological multimodality in the  posterior distribution, we subsample $n=1000$ training data from the labels $y=1$ and 2 and set the number of hidden nodes $J=40$. We use hierarchical priors, $\alpha  \sim $ normal(0,$\sigma_\alpha$), $\beta  \sim $ normal(0,$\sigma_\beta$), $\sigma_\alpha, \sigma_\beta \sim$ normal$^+(0,3)$. Switching the order of hidden nodes does not change the predictive density. We eliminate the combinatoric non-identification in all other experiments in this section  by constraining the order of $\beta$: $\beta_1 \geq \beta_2  \dots \geq \beta_J$. 
	
	\begin{figure}
		\includegraphics[width=1.05\linewidth]{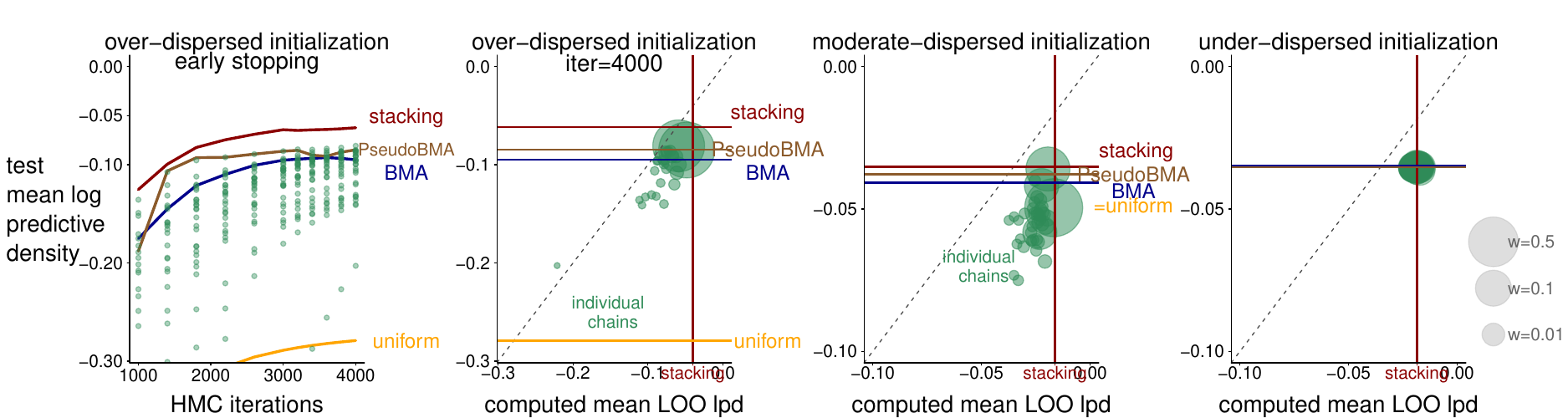}
		 	\caption{\em  (1): The test mean log predictive densities  of early stopped chains.  Stacking performs consistently better than single-chain or other weighting methods.
			(2)--(4):  The mean leave-one-out and test data log predictive densities of 50  individual chains (green dots), their stacking weights (size of the dot), and the test mean lpd from four weighting strategies when fitting a 40-hidden node neural network on MNIST.  There were 4000 iterations per chain, and network parameters  are  initialized  from uniform$(-50,50)$, $(-5,5)$, and $(-0.001,0.001)$, respectively. Some individual changes in the overdispersed setting are out of the lower range.
		}	\label{fig_nn_wo_id}
	\end{figure}
	
	We sample from the posterior distribution $p(\phi, \beta, \alpha | y, x)$ using 50 parallel HMC/NUTS chains in Stan.  The right three panels in Figure \ref{fig_nn_wo_id} show the posterior predictive performance of individual chains and combinations, evaluated by the mean log predictive densities on both leave-one-out data and   test data with $n_\mathrm{test}=2167$.   
	The test score standard deviation is negligible.  The initial values of  unconstrained parameters in panels 2--4 are drawn from uniform$(-50,50)$, $(-5,5)$, and $(-0.001,0.001)$, respectively. Each green dot stands for one chain, and the size of the dot reflects the chain weight in stacking (we rescale the size proportional to $w_\mathrm{stacking}^{1/5}$ to manifest extremely small weights, see the legend on the right). 
	Under an overdispersed initialization, the  posterior inferences considerably diverge, and uniform weighting is jeopardized by ``unlucky'' chains, while stacking is not affected  by a large number of bad chains.  The PSIS-loo approximation does not accurately estimate the test performance (detected by large $\hat k$  diagnostics), but stacking still outperforms all other weighting strategies.  Under the $(-0.001,0.001)$ initialization, all 50 chains are essentially identical, and there is no gain from reweighing.  In this experiment, a carefully tuned underdispersed initialization is the most efficient. However, choosing optimal starting values in general models remains difficult, whereas stacking is less sensitive to the initialization.
	
	Early stopping is a commonly used ad hoc regularization method in neural networks \citep{vehtari2000mcmc}. The leftmost column in Figure \ref{fig_nn_wo_id} demonstrates that we can stack early stopped chains to achieve a prediction-power and computation-cost tradeoff. In the setting of 40 hidden nodes and overdispersed initialization, stacking is strictly better than the best single chain, however early we stop.  Stacking with 1500 HMC iterations is better than the best chain at iteration 4000. BMA and pseudo-BMA effectively choose just a single chain, and they select the wrong chains at times. Uniform weighting is again the worst due to its sensitivity to bad initializations.  
	
	Some literature on neural net ensembles advocates to uniformly average over all ensembles constructed by local MAPs found through  stochastic gradient descent \citep{lakshminarayanan2017simple}, bootstrap resampling  \citep{osband2017deep}, or varying priors  \citep{pearce2018uncertainty}.  Our experimental results show that inference from uniform weighting can be highly sensitive to starting points and can be especially disappointing under an overdispersed initialization. The approximate loo-based stacking sheds light on the benefit of post-inference multi-chain-reweighing in modern deeper neural networks.  The additional optimization cost is tiny compared to the cost of model training.  We leave the           question of scalability to modern Bayesian deep learning models to future investigation.

	\section{Discussion}

\subsection{The Folk Theorem of Statistical Computing} 
\emph{When you have computational problems, often there’s a problem with your model}. This heuristic or ``folk theorem'' \citep{gelman2008folk} can be understood by thinking of a statistical model or family of distributions as a set of possible probabilistic explanations for a dataset.  If the data come from some distribution in the model class, then with identification and reasonable sample size we can expect to distinguish among these explanations, and with a small sample size and continuous model, we would hope to find a continuous range of plausible explanations and thus a well behaved posterior distribution. Indeed,  under correct models and reasonable priors, Bayesian posteriors often attain asymptotic normality and leave little room for distinct and non-vanishing modes. That ensures rapid mixing for random-walk Metropolis, scaling as $\mathcal{O}(d)$ \citep{roberts1997weak,cotter2013mcmc, dwivedi2018log},
and Hamiltonian Monte Carlo, scaling as $\mathcal{O}(d^{1/4})$ \citep{beskos2013optimal, bou2018coupling,  mangoubi2017rapid,  mangoubi2019mixing}.

If the data do not fit the model, so that none of the candidate explanations work, then the posterior distribution represents a mixture of the best of bad choices, and it can have poor geometry in the same way that the seafloor can look rough if the ocean is drained.
Poor data fit, or conflict between the prior and likelihood, do not necessarily lead to awkward computation.  For example, the normal-normal model yields a log-concave posterior density with constant curvature for any data. But if a model is flexible enough to fit different qualitative explanations of data, then poorly fitting data can be interpreted by the model as ambiguity, as indicated by posterior multimodality.

The other way a model can be difficult to fit is if its parameters are only weakly constrained by the posterior.  With a small sample size (or, in a hierarchical model, a small number of groups), uncertainty in the hyperparameters can yield a posterior distribution of widely varying curvature, which leads to slowly mixing MCMC.  In practice, we can often reshape the geometry by putting stronger priors on these hyperparameters. 
However, a strong prior constraint is not always desired---sometimes we are interested in fitting a model that is legitimately difficult to compute, because we want to allow for different possible explanations of the data, and a too strong prior implies an ad hoc selection.  These are settings where the stacking approach discussed in this paper can be useful.

	\subsection{Learn Better Epistemic Uncertainty to Expiate Aleatoric Misspecification}
	Uncertainty comes into inference and prediction through two sources: (a) due to finite amount of data, we learn the \emph{epistemic} uncertainty of  unknown parameter $\theta$ through the posterior distribution $p(\theta|y)$, and (b)  due to either the stochastic nature of real world, even when $\theta$ is known,  we represent the \emph{aleatoric} uncertainty through the probabilistic forecast of next unseen outcome as $p(\tilde y|\theta, y)$. The final probabilistic prediction contains both of them via $p(\tilde y|y)= \int\! p(\tilde y|\theta, y) p(\theta|y) d \theta$.
	
	Given a model, the  epistemic uncertainty is mathematically well-defined through Bayesian inference, but will only  be optimal  under the true model and when averaging over the prior distribution.   By being open-minded  to  model misspecification, the  optimization \eqref{eq_limiting_goal}  searches for the ``best'' probabilistic inference and uncertainty quantification with respect to a given utility function.
	
	Our paper calls attention to  postprocessing and calibrating Bayesian epistemic uncertainty.   Stacking reweighs separated components in the posterior density,  while in general we can consider other  transformations of the posterior draws such as location–scale shift,   mixtures, and convolutions. 
	
	Bayesian inference is known to be poorly calibrated under model misspecification.  In the context of model-selection and averaging, the marginal-likelihood-based ``full-Bayes'' approach produces
	over-confident prediction when none of the models is true \citep{clarke2003comparing, wong2004improvement, clyde2013bayesian, yao2018using, yang2018good,  oelrich2020bayesian}, and therefore is not  Bayes optimal \citep{le2017bayes}.   
	
	The suboptimality of Bayesian posteriors does not mean we think Bayesian inference is wrong,  but it does imply that there are tensions between a reckless application of Bayes rule under the wrong model and the Bayesian decision theory,  and more generally, between Bayesian inference and Bayesian workflow.
	In the words of  \citet{gelman2020holes},  such tensions can only be resolved by considering Bayesian logic as a tool, a way of revealing inevitable misfits and incoherences in our model assumptions, rather than as an end in itself.

	\subsection{Stacking as Part of    Bayesian Workflow}
	We view stacking of parallel chains as sitting on the boundary between black-box inference and a larger Bayesian workflow \citep{ BayesianWorkflow2020}.  
	
	For an automatic inference algorithm, stacking enables accessible inference from non-mixing chains and a free enrichment of predictive distributions, which is especially relevant for repeated tasks where computation time is constrained.
	
	For Bayesian workflow more generally, we recommend stacking in the model exploration phase, where we need to obtain \emph{some} inference.  Parallel computation can be running asynchronously---it may be that only some chains are running slowly---and stopping in the middle frees up computation and human time that can be reallocated to explorations of more models.  In addition, non-uniform stacking weights when used in concert with trace plots and other diagnostic tools can help us understand where to focus that effort in an iterative way.

\acks{We thank the U.S. National Science Foundation, Institute of Education Sciences, Office of Naval Research, Sloan Foundation, and the Academy of Finland Flagship programme: Finnish Center for Artificial Intelligence, FCAI, for partial support of this work.}

	\small
	\bibliography{multimodal}

	\newpage
	\renewcommand\thesection{\Alph{section}}
	\normalfont
	\normalsize
	\section*{Appendices}
	\setcounter{section}{0}
	\section{Proofs for asymptotic theories}
		We sketch the proof for theorems in Section \ref{sec_theory}.
	\subsection{Proofs for Corollary  1}
	
	The proof of Corollary 1 is a direct application of the consistency results in  \citet{vehtari2015pareto} and \citet{le2017bayes}.
	
	Assuming samples from the $k$-th chain ($k=1,2, \dots, K$)  (not necessarily independently) come from a stationary distribution $p_k(\theta)$,we denote $p_{k,  -i}(y_i) =  p_k(y_i| y_{-i}) \coloneqq \int_{\Theta} p(y_i|\theta)p_k(\theta| y_{-i})d\theta $ to be the leave-one-out density. 
	
	First, the importance sampling based approximation is pointwise consistent.
	\begin{theorem} \citep[Theorem 2 and 3 in][]{vehtari2015pareto}  Assuming the stationary distribution $p_k(\theta)$ satisfies  regularity conditions defined therein, the PSIS-based approximate loo is consistent with a large number of posterior draws. For any fixed chain index $k$, and observation index $i$,
		$$  \frac{\sum_{s=1}^S p(y_i|\theta_{ks})r_{iks}}{\sum_{s=1}^S r_{iks}}  - p_{k,  -i}(y_i)  \xlongrightarrow{\text{$L_2$}}    0,  ~ S  \to \infty.$$
	\end{theorem}
	
	In practice, the convergence rate of approximate PSIS-loo with finite posterior draws can be characterized by the $\hat k$ diagnostics \citep{vehtari2015pareto}.

	Second, \citet{le2017bayes} proved that given set of weights $w_1 \dots w_K$ and   when sample size $n \to \infty$, the  leave-one-out  logarithmic predictive density (loo lpd),  converges to the expected log predictive densities (elpd): 
	
	\begin{theorem} 	\citep[Theorem 2.2  in][]{le2017bayes}  Assuming regularity conditions:
		\begin{enumerate}
			\item  For each $k = 1, \dots,K$, there is a function $B_k(\cdot)$ so that $$\sup_{y\in \R^n}  | \log p_k (\tilde y | y)    | \leq  B_k(\tilde y) < \infty,  $$
			where $B_k$ is  independent of other covariates and $\E( g(\tilde y) ) < \infty$  for 
			$$g(\tilde y )   = \max\left\lbrace      \left(        \log  \sum_{k=1}^K  w_k  \exp\left(-B_k (\tilde y)\right)          \right)^4 ,    \left(        \log  \sum_{k=1}^K  w_k  \exp\left(B_k (\tilde y)\right)          \right)^4   \right\rbrace.  $$
			\item  For each $k = 1, \dots,K$, the conditional densities $p_k (y | x,  \theta )$ are equicontinuous in $x$ for each $y$ and $\theta \in \Theta_k$, and the predictive densities $p_k(\cdot| y )$ within the are uniformly equicontinuous in y.
		\end{enumerate}
		
		Then we have 	
		$$
		\frac{1}{n}\sum_{i=1}^n  \log \sum_{k=1}^K  w_k  p_{k,  -i}   (y_i ) -  \mathrm{E}_{\tilde y| y}  \log \sum_{k=1}^K w_k p_k(\tilde y| y )          \xlongrightarrow{\text{$L_2$}}   0,  ~ n\to \infty.
		$$
	\end{theorem} 
	
	Now return to the objective function in stacking (Equation \ref{eq_stacking_objective}):
	$$
	\max_{\w\in \mathbbm{S}(K)} \sum_{i=1}^n \log \sum_{k=1}^K w_k p^S_{k, -i}(y_i) + \log p_{\mathrm{prior}}(\w),$$
	where the leave-one-out distribution is approximated by importance sampling using $S$ posterior draws each chain,
	$$p^S_{k, -i}(y_i)=   \frac{ \sum_{s=1}^{S} p_k(y_i|\theta_{ks}) r_{iks}}{   \sum_{s=1}^{S}  {r_{iks} }}.$$
	
	Combining the previous two consistency results, for a fixed number of chains $K$ and a fixed weight vector $\w$,  when both the sample size of observations $n$ and the number of posterior draws $S$ go to infinity,  under all previous mentioned assumptions, the objective function converges to the elpd of the weighted posterior inference:
	$$ \frac{1}{n} \sum_{i=1}^n \log \sum_{k=1}^K w_k p^S_{k, -i}(y_i)  -  \mathrm{E}_{\tilde y| y}  \log\Bigl( \sum_{k=1}^K w_k p_k(\tilde y| y )  \Bigr)    \xlongrightarrow{\text{$L_2$}}   0,$$
	which proves Theorem \ref{them_op}.
	
	\subsection{Proofs of Theorems 2 and 3}
    First, the unnormalized posterior density of $\mu$  is 
	$$\log p(\mu|y) =  \log p_0(\mu) - \sum_{i=1}^n  \log(1+  (y_i-\mu)^2).    $$
	Define $$h(\mu)=  -   \int_{-\infty}^{\infty} \log(1+  (y-\mu)^2)  \left(\frac{1-p_0}{(a+y)^2+1}+\frac{p_0}{(y-a)^2+1}\right)    d y,$$   which is always a well-defined and finite integral for all $\mu$.   
	\begin{lemma}\label{lem_h_d}
		$\frac{d}{d\mu} h(\mu)$ has a closed form expression
		\begin{align*}
		\frac{d}{d\mu} h(\mu)	&=-\int_{-\infty}^{\infty}      \frac{d}{d\mu}   \log(1+  (y-\mu)^2)  \left(\frac{1-p_0}{(a+y)^2+1}+\frac{p_0}{(y-a)^2+1}\right)    d y   \\
		&=-\frac{\pi  p_0 (\mu-a)}{(a-\mu)^2+4}-\frac{\pi  (1-p_0) (a+\mu)}{(a+\mu)^2+4}\\
		&=\frac{-\pi  \left( 4 a + a^3 - 8 a p - 2 a^3 p + (4 - a^2) u + (-a + 2 a p) u^2 + u^3  \right)}{\left(a^2-2 a \mu+\mu^2+4\right) \left(a^2+2 a \mu+\mu^2+4\right)}.
		\end{align*}
	\end{lemma}		
	\begin{proof}	
		Calculus and change of variables.
	\end{proof}	
	We define  $\xi(a)$  as the third largest root of the following forth-order equation (as a function  of $x$):
	$$ u_a(x) = x^4 \left(a^6+4a^4\right)+x^3 \left(-2 a^6-8 a^4\right)+x^2 \left(a^6-8 a^4-44 a^2\right)+x \left(12 a^4+44a^2\right)-4 a^4-8 a^2-4 =0.$$
	$\xi(a)$  is a bijective and increasing mapping from $[2, \infty)$ to $[0.5, 1)$. $\xi(2) =0.5$ and $\lim_{a\to \infty} \xi(a) =1$. 
	We visualize the deterministic function $p_0=\xi(a)$ in Figure \ref{fig_xi_function}.
	\begin{lemma} \label{lemma_h}	
		The number of modes in $h(\mu)$ is determined by the relation between $a$ and $p_0$.  
		
		(a) When $a>2$ and $ p_0 \geq \xi(a)$, $h(\mu)$ only has one global maximum  near $a$. 
		
		(b) When $a>2$ and $p_0 < \xi(a)$, $h(\mu)$ has two local maximum  near $a$ and $-a$ respectively. 
		
		(c) When $a<2$, $h(\mu)$ is unimodal with the global maximum between 0 and a.	
	\end{lemma}		
	
	\begin{proof}	
		
		The denominator in $\frac{d}{d\mu} h(\mu)$ is always positive.     Let $g(\mu)=-4 a - a^3 + 8 a p + 2 a^3 p + (-4 + a^2) u + (a - 2 a p) u^2 - u^3$.  It is a cubic polynomial on $\mu$ and has the discriminant:
		\begin{align*}
		\Delta(a, p_0)=& \left(64 a^6+256 a^4\right) p_0^4+\left(-128 a^6-512 a^4\right) p_0^3+ \left(64 a^6-512 a^4-2816 a^2\right) p_0^2\\ &+\left(768 a^4+2816 a^2\right) p_0-256 a^4-512 a^2-256.
		\end{align*}
		Solving 	$ \Delta(a, p_0)=0$ has and only has one root on $a>2$ and $0.5<p_0<1$:
		$p_0=\xi(a)$, where the function $\xi(a)$ is defined in the lemma.
		
		Further, when $p_0 \geq \xi(a)$,  $\Delta(a, p_0)\leq   0$, and therefore $g(\mu)$ only has one cross-zero-root.
		Since  $h'(\mu)=g(\mu)$ and  $h(\pm \infty)=-\infty$, this unique root is the global maximum  of $h(\mu)$. We denote this unique mode by $\gamma(a, p_0)$.
		
		For a large $a$, using the second expression in Lemma \ref{lem_h_d}, 	$\frac{d}{d\mu} h(\mu)|_{\mu=a}  =  - \frac{\pi  a(1-p_0) }{2a^2+2}\to 0^-$. Therefore the mode $\gamma(a, p_0) \to a^-,$ as $a\to \infty$.
		
		In situation (b), when $p_0 < \xi(a)$,  $\Delta(a, p_0) >   0$. $g(\mu)$ only has three cross-zero roots. This implies $h(\mu)$ has two local maxima  $\gamma^+$ and $\gamma^-$, near but not identical to $\pm a$, and a local minimal (near 0). 
		
		Using the second line in Lemma \ref{lem_h_d}, for any $\mu<0$, $h(-\mu)> h(\mu)$, therefore $h(\gamma^+)>   h(-\gamma^-)>  h(\gamma^-)$; that is, the right mode is higher than the left mode for $p_0> 0.5$.
		
		In situation (c), when $0<a<2$, $\Delta(a, p_0) <  0$ and   therefore $g(\mu)$ only has a cross-zero-root, which is the first root in the following cubic function: 
		$$u(x)= x^3 +x^2 (a p_0- a) + x \left(4 - a^2\right)- 2a^3 p_0+a^3-8ap_0+4a =0. $$
		In particular if $p_0=0.5$, this root is at $\mu=0$.
	\end{proof}

	\begin{lemma} \label{lemma_p5}
		For a fixed $p_0$ and $a\to \infty$, the two local modes  $\left( \gamma^+ (a, p_0), \gamma^-(a, p_0)\right)  \to (a, -a)$.
	\end{lemma}		
	\begin{proof}		
		Using the second expression in Lemma \ref{lem_h_d}, 	$$\frac{d}{d\mu} h(\mu)|_{\mu=a}  =  - \frac{\pi  a(1-p_0) }{2a^2+2}\to 0^-, ~\mathrm{as}~ a \to \infty,$$ while $$\frac{d^2}{d\mu^2} h(\mu)|_{\mu=a} = \frac{\pi  a (8 a p_0-8 a)}{\left(4 a^2+4\right)^2}-\frac{\pi  \left(-4 a^2 p_0-4\right)}{4 \left(4 a^2+4\right)} \to- \frac{\pi p_0}{4}= O(1).$$ Hence when $a\to \infty$. the mode $\gamma^+(a, p_0) \to a^-$, and likewise $\gamma^-(a, p_0) \to -a^+$.
	\end{proof}
	
	The approximation using Lemma \ref{lemma_p5} is accurate for a moderately large $a$.
	For example, when $p_0 =0.6$, and $a=8$, the  right and left modes in $h$ are $(\gamma^+ (a, p_0), \gamma^-(a, p_0))=(7.8, -7.3)$, and at $a$=10 they are $(9.9, -9.7)$.

	\begin{lemma}   When $a>2,  p_0 = 0.5 $,   $h(\mu)$  has two equally high modes at $\pm\, \sqrt{a^2-4}$.
	\end{lemma}		
	
	\begin{proof}
		This is a special case of the previous lemma in which we can  solve   $h'(\mu)=0$ explicitly. 
		$$  	\frac{d}{d\mu} h(\mu |a, p_0 = 0.5)=-\frac{2 \pi  \mu \left(-a^2+\mu ^2+4\right)}{-2 a^2 \left(\mu^2-4 \right)+a^4+\left(\mu^2+4\right)^2}$$
		has three zeros,  0 and $\pm \mu_0$, where $\mu_0=\sqrt{a^2-4}$.  Furthermore we can check $h''(0)>0$, and $h''(\pm \mu_0) <0$.  Hence $h(\mu)$ has one local minimal at $\mu=0$ and two global maximum at  $\pm \mu_0.$   $h(\mu_0)= h(-\mu_0)$ due to symmetry.
	\end{proof}
	
	When $a>2,  p_0 >0.5$,  $h(\mu)$ either has a unique mode ((a) in Lemma \ref{lemma_h}), $\gamma^+>0$, or two local modes  ((b) in Lemma \ref{lemma_h}) with unequal heights $h(\gamma^+)> h( \gamma^-)$. The convergence to the right mode is a straightforward application of any usual Bayes consistency result (under model misspecification).
	
	\begin{lemma}   When $a>2,  p_0 >0.5 $,    the posterior $p(\mu | y_1, \dots, y_n)$ is asymptotically concentrated at the point mass $\gamma_+$. That is, for any $\eta>0$,   when $n\to \infty$,  $$\Pr(   \left| \mu - \gamma^+   \right | < \eta   |y_1,\dots, y_n)    \to  1,  a.s. $$
	\end{lemma}
	
	\begin{proof}
		The weak law of large numbers implies 
		$$\frac{1}{n} \log C_n p(\mu|y_1,\dots, y_n)    \to  h(\mu), $$
		where $C_n$ is the normalization constant.
		Since $h$ is $C^{\infty}$ smooth, we can choose $\delta = \frac{1}{2}\left(  h(\gamma^+)- h( \gamma^-)\right)>0 $, and there exists an $\epsilon$ neighborhood of $\gamma_+$ such that,   $$\inf_{\gamma:|\gamma- \gamma^+|<\epsilon}  h(\gamma)> h(\gamma^+)  -\delta > \sup_{\gamma: |\gamma- \gamma^+|>\epsilon} h(\gamma),$$  
		which implies  $$\Pr(\mu \in   (\gamma^+ - \epsilon ,   \gamma^+ + \epsilon)    |y_1,\dots, y_n)    \to  1$$
		
	\end{proof}	
	
	Now express the log posterior density of $\mu$ as
	
	\begin{align*}
	\log p(\mu|y_{1, \dots, n}  ) &=   \log  p_0 (\mu)+  \sum_{i=1}^n   -\log(1+  (y_i-\mu)^2)  - \log C_n   \\ 
	&= \log p_0 (\mu)+  n h(\mu) + \sqrt n G_n (\mu) -  \log C_n,
	\end{align*}
	where   $\log C_n$ is the log normalization constant, and 
	$$G_n (\mu) =  n^{-1/2} \sum_{i=1}^n  \left(     -\log(1+  (y_i-\mu)^2)  - h(\mu)       \right), $$
	which can also be written as
	$$G_n (\mu) = \int\! -\log (1-(\mu-y)^2) dB_n (y),  \quad B_n (y)= \sqrt n   ( {F}_n - {F}  ). $$
	where   ${F}_n$ and  ${F} $ are the empirical distribution of $y_1, \dots, y_n$ and the distribution function of the data generating process, respectively.
	
	The remaining argument transfers the results from $h(\mu)$ to the posterior. Loosely speaking, the remaining term $G_n (\mu)$ is asymptotically a Gaussian process and bounded by $o(n^{1/2})$, while the main term $nh(\mu)$ outside the neighborhood of the mode of $h(\mu)$ vanishes  $O(n)$ quicker than the inside. Therefore, the posterior $p(\mu| y_{1:n})$ will  asymptotically carry a mode around the mode in $h(\mu)$. That is Theorem \ref{thm_post_mode}.  A rigorous proof of Theorem \ref{thm_post_converge} follows from all previous lemmas and Lemma 2.4-2.12 in \citet{diaconis1986inconsistent}.
	
	\subsection{Proofs for Corollaries 4 and 5}
	Corollary 4 follows directly from Theorem 3.  
	In specific, for big $a$, we can further approximate the left and right mode near $\pm a$ using Lemma \ref{lemma_p5}. Then the Bayesian posterior is closed to a point mass that is spiked at $a$ for $ 0.5< p_0 < \xi(a) $, so the resulting KL divergence is always non-vanishing. The KL divergence between two Cauchy densities  $\mathrm{Cauchy} (\mu_1, \sigma)$ and  $\mathrm{Cauchy} (\mu_2, \sigma)$ has a closed form expression:
	$\KL\bigl( \mathrm{Cauchy} (\mu_1, \sigma)  ~||~ \mathrm{Cauchy} (\mu_2, \sigma)\bigr)  = \log \left(1 + \frac{(\mu_1-\mu_2)}{4\sigma^2}  \right).$
	
	In Corollary 5, we assume the parallel evaluation has captured both modes $\gamma^-$  and  $\gamma^+$ and we have classified them into two clusters. Using Corollary 1 , for any $0.5 < p_0 < \xi(a)$,  stacking solves 
	$$\min_{w \in \mathbbm{S}(2)} \KL\Bigl( (1-p_0)\, \mathrm{Cauchy} (a, 1)  +  p_0 \,\mathrm{Cauchy} (-a, 1)  ~|| ~w_1\, \mathrm{Cauchy} (\gamma^-, 1) + w_2 \, \mathrm{Cauchy} (\gamma^+, 1) \Bigr)  .$$
	The limiting Bayesian inference is a stacking solution corresponding to a weight of 1 on the right mode.  It is easy to check that $w=(0,1)$ is not the optimum by first order conditions. Using Corollary 1 we see the stacking weights yields a higher elpd.
	
	When $p_0=0.5,~ a>2$, in the $n \to \infty$ limit in Corollary \ref{them_op}, the stacking solution optimizes
$\min_{w \in \mathbbm{S}(2)} \KL\bigl( 0.5\, \mathrm{Cauchy} (a, 1)  +  0.5\, \mathrm{Cauchy} (-a, 1)~ ||~  w_1 \,\mathrm{Cauchy} (\sqrt{a^2-4}, 1) + w_2\,  \mathrm{Cauchy} (-\sqrt{a^2-4}, 1) \bigr), $
	which is attained at $w_1=w_2=0.5$.
	Direct computation shows tjat the KL divergence above at the optimal  $w_1=w_2=0.5$ approaches $0$ for big $a$. See Figure  \ref{fig_xi_function} for numerical evaluations.
	
	\section{Implementation in \texttt{Stan} and \texttt{R} package \texttt{loo}}
	We demonstrate the implementation of multiple-chain stacking in the general-purpose Bayesian inference engine Stan \citep{stan2020}.	We use the Cauchy mixture model as an example.  First save the following Stan file to \texttt{cauchy.stan}.
\begin{verbatim}
data {
  int n;
  vector[n] y;
}
parameters {
  real mu;
}
model {
  y ~ cauchy(mu, 1); 
}
generated quantities {   
  vector[n] log_lik;
  for (i in 1:n)
  log_lik[i] = cauchy_lpdf(y[i]| mu, 1);
}
\end{verbatim}
In the \texttt{generated quantities} block, we save \texttt{log\_lik}: the log likelihood of each data point at each posterior draw.   We generate data from a Cauchy mixture according to example (iii) in Figure \ref{fig_mode1}, and sample from its posterior densities.  Here is the R code:
	\begin{verbatim}
	library(rstan)
	library(loo)
	set.seed(100) 
	mu = c(-10,10)
	n = 100
	y = rep(NA, n)
	p = 0.5
	y[1:(n*p)] = rcauchy(n*(p),mu[1], 1)
	y[(n*(p)+1):n] = rcauchy(n*(1-p),mu[2], 1)
	K = 8
	# Fit the model in stan
	set.seed(100)
	stan_fit = stan("cauchy.stan", data=list(n=n, y=y), chains=K, seed=100)
	mu_sample = extract(stan_fit, permuted=FALSE, pars="mu")[,,"mu"]
	print(Rhat(mu_sample))
	\end{verbatim}
	 We are using eight parallel chains, and the resulting  $\widehat R$ is 1.6, indicating clear problems with mixing.
	
The R function \texttt{chain\_stack()} combines multiple  chains in a Stan fit object, returned by \texttt{stan()}. It only require the whole model fit once, and save the point wise log likelihood in each iteration,  called via \texttt{log\_lik} here. The \texttt{chain\_stack()} function uses the Stan optimizer (the default is L-BFGS), and its first time compiling  takes up to a few minutes. The tuning parameter \texttt{lambda} controls the Dirichlet prior on stacking weights.
	\begin{verbatim}
	> library(devtools)
	> source_url("https://github.com/yao-yl/Multimodal-stacking-code
	/blob/master/chain_stacking.R?raw=TRUE")
	> stan_model_object = stan_model("stacking_opt.stan")
	> stack_obj=chain_stack(fits=stan_fit,lambda=1.0001,log_lik_char="log_lik")
	
	Output:  Stacking 8 chains, with 100 data points and 1000 posterior draws;
	using stan optimizer, max iterations = 1e+05 
	...done.
	Total elapsed time for approximate LOO and stacking = 0.87 s 
	\end{verbatim}
We can assess the reliability of the approximate leave-one-out using the $\hat k$ diagnostics. In this example, all pointwise  $\hat k$ estimates (100 observations $\times$ 8 chains = 800 in total) are smaller than 0.5, indicating that the loo approximation is accurate in this example.
	\begin{verbatim}
	>  print_k(stack_obj)
	
	Output:                Count Proportion
	(-Inf, 0.5] (good)     800   1         
	(0.5, 0.7]  (ok)       0     0         
	(0.7, 1]    (bad)      0     0         
	(1, Inf)    (very bad) 0     0 
	\end{verbatim}
	We  access the chain wights using
	\begin{verbatim}
	> chain_weights = stack_obj$chain_weights
	\end{verbatim}
Finally, we can use  the weighted samples to calculate any posterior integral  $\E_{\mathrm{stacking}}(h(\mu) | y )$ as in  \eqref{eq_MC_final}. Here we compute $\Pr(\mu>0 | y)$:  the total  mass  of positive values in the stacked inference.
	\begin{verbatim}
	> h = function(mu){mu>0}
	> round(chain_weights %*% apply(h(mu_sample), 2, mean), digits=3)
	[1] 0.523
	\end{verbatim}
Alternatively, we provide a quasi Monte Carlo based importance resampling function  \texttt{mixture\_draws()}  that  draws posterior samples form the stacked inference. This enables us to compute the same integral $\E_{\mathrm{stacking}}[h(\mu) \mid y]$ using usual Monte Carlo methods:
	\begin{verbatim}
	> resampling=mixture_draws(individual_draws=mu_sample,weight=chain_weights)
	> mean(h(resampling)) 
	[1] 0.523
	\end{verbatim}

	\section{ Reproducible code and experiment details}
	Data and code for this paper are available at \url{https://github.com/yao-yl/Multimodal-stacking-code}.
	\paragraph{LDA topic models.}
	In Section \ref{sec_lda}, the text data are all words in the novel Pride and Prejudice. We preprocess the data by removing stop words and rare words.  The cleaned data are stored in  the posterior database (\url{https://github.com/MansMeg/posteriordb}), also uploaded as \texttt{staninput.RData}.
	We use the Stan implementation of LDA models (\url{https://mc-stan.org/docs/2_22/stan-users-guide/latent-dirichlet-allocation.html}) with little modification, as in the file \texttt{lda.stan}.
	
	In all experiments, We run parallel  inference on Columbia University's shared HPC Terremoto with one chain per core (CPU: Intel Xeon Gold 6126, 2.6 Ghz). When there is no further specification, we use the default starting values: draw all unconstrained parameters from uniform$(-2,2)$ randomly in each chain. 
	
	We pre-specify the maximum running time for 2000 iterations to be 24 hours and 4000  iterations to be 48 hours in all LDA models, and all running-out-of-time chains are discarded.
	
	\paragraph{Gaussian process regression.}
	The original data of \citet{neal1998regression} can be found in file \texttt{odata.txt}. In the first experiment, we use the first half as training data. In the second experiment, we simulate data with varying sample size according to his data generating process.  
	For hyper-parameter optimization, we found  two modes  by using initialization $(\log \rho, \log \alpha, \log \sigma)= (1, 0.7, 0.1)$ and $(-1, -5, 2)$, respectively. We approximate the posterior by MAP or Laplace approximation and importance resampling around two local mode.  The approximate samples have little overlap.
	
	In the full sampling for the $t$ regression, we compare four chain-combination strategies:  BMA, pseudo-BMA, uniform averaging, and stacking.  After each iteration of $(\sigma, \rho, \alpha, f)$, we draw  posterior predictive sample of $\tilde f= f(\tilde X)$,  from  $$\tilde f | \tilde X, X, f  ~ \sim \mbox{MVN}\left(K(\tilde X, X) K(X, X)^{-1} f, K(\tilde X, \tilde{X})- K(\tilde X, X)K(X, X^{-1})K(X, \tilde X)   \right), $$
	and compute the mean test data log predictive densities,
	$$1/n_\mathrm{test} \sum_{i=1}^{n_\mathrm{test}} \log p(\tilde y_i | \tilde f_i, \sigma) p(\tilde f_i, \sigma | X, y)d \tilde f_i d \sigma.$$
	The full-model specification is in \texttt{treg.stan}.
	

	\paragraph{Balanced one-way hierarchical model.}
	There can be entropic barriers in the non-centered parameterization too. The likelihood in \eqref{eq_ncp} is equivalent to $\xi_i \vert \tau, \mu, y  \sim \mbox{normal}(\frac{1}{\tau}(\bar y_{.j} - \mu) ,  \sigma  \tau^{-1} J^{-1/2} )$, where $\bar y_{.j}$ is the sample mean of group $j$.  Replacing $\tau$  and $\theta_j$ by plug-in estimates, we derive the conditional variance in the likelihood as $\mathrm{Var}(\xi_i | \mu, \sigma, y) \approx  \left( N^{-1} J \sigma^2\right) /  \sum_j (\bar{y_{.j}}-\mu)^2,$  which forms a funnel between $\mu$ and $\xi$.
	
	In the experiment, the true $\tau$ and $\sigma$ vary from $0.1$ to $20$.  In order to achieve a higher F-statistics so as to manifest posterior bimodality, we   additionally add  some student $t$-distributed noise added to  group mean in the unknown data generating process.
	$\theta_i  := \theta_i +   B  z_i    $, where $z_i$ is iid $t(1)$ distributed noise, and B varies from 0 to 50. The complete pooling, centered, and non-centered parameterizations are coded in the Stan files \texttt{random-effect-zero.stan},  \texttt{random-effect.stan} and \texttt{random-effect-ncp.stan}.
	
	\paragraph{Neural networks for MNIST.}
	\begin{wrapfigure}{R}{0.5\textwidth}
		\centering
		\includegraphics[width= 1 \linewidth]{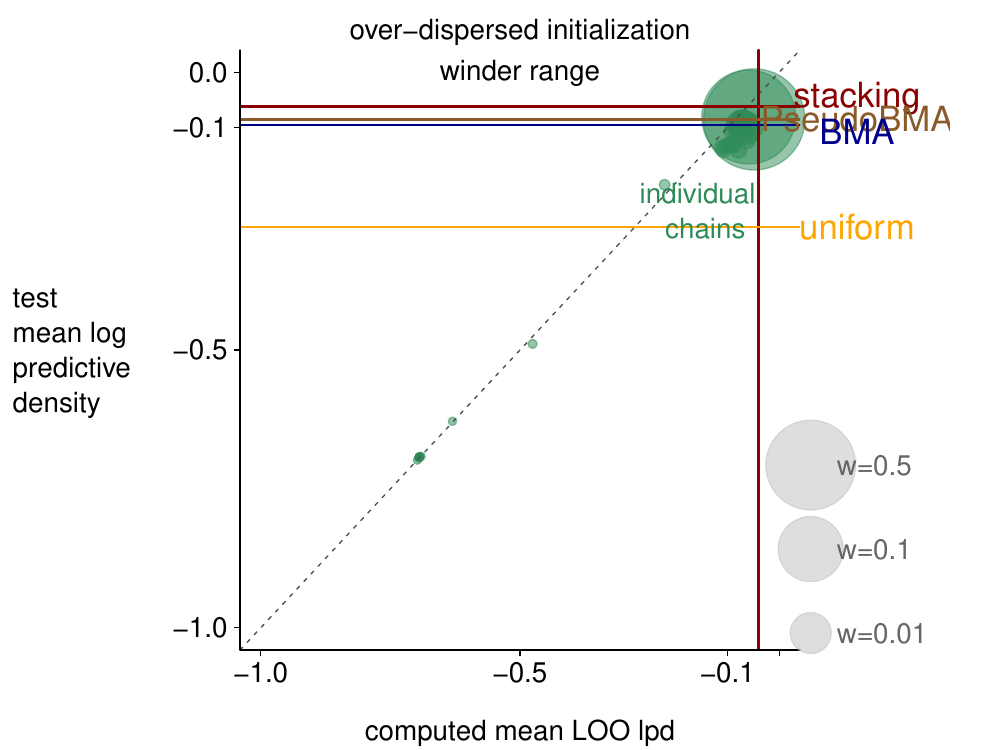}
		 	\caption{\em   Some individual changes in the overdispersed setting are out of lower-range and not shown in Figure \ref{fig_nn_wo_id}. This is the same graph with wider ranges.	} 
	\end{wrapfigure}
	We subsample 1000 data points from MNIST as training data, with subsampling details in \texttt{readmnist.R} and the saved test and training data in \texttt{input.RData}. The model is adapted from  Bob Carpenter's Stan code \url{https://github.com/stan-dev/example-models/blob/master/knitr/neural-nets/nn-simple.stan} with a few modifications as in \texttt{2classnn.stan}.
	
	In the experiment, we considered two choices of priors: (a) a fixed-scale elementwise normal$(0,3)$ prior on all unknown parameters $\phi \in \R, \beta \in \R^{40}$, and  $\alpha\in \R^{784\times 40}$; and (b)   $\alpha  \sim \mbox{normal}(0,\sigma_\alpha)$, $\beta  \sim \mbox{normal}(0,\sigma_\beta)$, $\sigma_\alpha, \sigma_\beta \sim \mbox{normal}^+(0,3)$.  For the experiment we are running, these two sets of priors yield nearly identical posterior sampling results and the same results after chain averaging.

\end{document}